\documentclass{amsproc}

\usepackage[T2A]{fontenc}
\usepackage[cp1251]{inputenc}
\usepackage[english]{babel}

\usepackage{amsfonts}
\usepackage{amsthm}
\usepackage{amssymb}
\usepackage{amsmath}
\usepackage{amsopn}
\usepackage{bm}
\usepackage{array}
\usepackage{color}

\newtheorem{theo}{Theorem}

\theoremstyle{definition}
\newtheorem*{defin}{Definition}
\newtheorem{exam}{Example}
\newtheorem{rem}{Remark}

\usepackage{cite}
\allowdisplaybreaks[4]

\newcommand{\tens}[1]{\textsf{#1}}

\DeclareMathOperator{\Complex}{\mathbb{C}}
\DeclareMathOperator{\Natural}{\mathbb{N}}

\DeclareMathOperator{\sn}{sn}
\DeclareMathOperator{\cn}{cn}
\DeclareMathOperator{\dn}{dn}
\DeclareMathOperator{\wgt}{wgt}
\DeclareMathOperator{\res}{res}
\newcommand{\Jac}{\mathfrak{J}}
\newcommand{\rmd}{\mathrm{d}}
\newcommand{\rmi}{\mathrm{i}}
\newcommand{\mFr}{\mathfrak{m}}
\newcommand{\rFr}{\mathfrak{r}}
\newcommand{\wFr}{\mathfrak{w}}
\newcommand{\M}{\mathcal{M}}
\newcommand*{\qede}{\null\nobreak\hfill\ensuremath{\square}}

\title[Solution of the Jacobi inversion problem]{Solution of the Jacobi inversion problem \\
on non-hyperelliptic curves}
\author{Julia Bernatska}
\address{University of Connecticut, Department of Mathematics}
\email{jbernatska@gmail.com}

\author{Dmitry Leykin}
\address{Retired}
\email{dmitry.leykin@gmail.com }

\begin{document}
 
\maketitle

\begin{abstract}
In this paper we propose a method of solving the Jacobi inversion problem in terms of multiply periodic $\wp$ functions,
also called Kleinian $\wp$ functions. 
This result is based on the recently developed theory
of multivariable sigma functions for $(n,s)$-curves. Considering $(n,s)$-curves as  
canonical representatives in the corresponding classes of bi-rationally equivalent
plane algebraic curves, we claim that the Jacobi inversion problem on
plane algebraic curves is solved completely. 
Explicit solutions on trigonal, tetragonal and pentagonal curves are given as an illustration.
\end{abstract}

\section{Introduction}
The Jacobi inversion problem, that is finding a preimage of Abel's map,
plays a central role in the theory of Riemann surfaces, 
and finds important applications in the solution of integrable systems, see for example \cite{bbeim1994}.

Solving the Jacobi inversion problem on elliptic curves led to integration of 
the Euler top and the Lagrange top, \cite{audin}. A solution to the Euler top was accomplished in terms of ratios of
theta functions, called later the Jacobi elliptic functions $\sn$, $\cn$, $\dn$.
A solution to the Lagrange top was given in terms of the Weierstrass elliptic functions $\wp$ and $\wp'$.

In the Weierstrass approach, elliptic functions are derived from an entire function
called the sigma function $\sigma$ and defined by a series which satisfies a heat equation, 
see \cite[Eqs (7.)--(9.) Art.\;5]{weier}.
The concept of the sigma function was extended to hyperelliptic curves  by Klein in \cite{klein}.
However, no progress was made in constructing sigma functions 
in higher genera until the end of the 20-th century.

An explicit solution of the Jacobi inversion problem on
hyperelliptic curves was known in the end of 19th century, and can be found in \cite[Art.\;216]{bakerAF}.
A solution in genus $2$, provided that the curve has a canonical form: $y^2 = \lambda_0 + \lambda_1 x 
+ \lambda_2 x^2 + \lambda_3 x^3 + \lambda_4 x^4 + 4 x^5$,
is discussed in detail in \cite[Art.\;11]{bakerMPF}. 
This solution is given in terms of 
multiply periodic $\wp$ functions obtained from the entire function $\vartheta(u)=e^{au^2} \theta(v)$,
where $v=(v_1,v_2)$ and $u=(u_1,u_2)$ are integrals of the first kind normalized and not normalized, respectively,
 $au^2$ denotes a quadratic function of $u$ with a $2\times 2$ symmetric matrix $a$,
 and $\theta$ denotes the Riemann theta function.
In fact, $\vartheta$ served as the sigma function related to the canonical form
of the class of genus two hyperelliptic curves, and renamed $\sigma$ in \cite[Art.\;12]{bakerMPF}. 
However, a series representation of  hyperelliptic sigma function (in any genus) was not known at that time.

This solution was rediscovered in 1990's, see  \cite{belHKF} and some further publications.
At the same time a development of the theory of multivariable sigma functions started.
First, the sigma function was defined through the Riemann theta function, as shown above.
In the beginning of 21st century a great progress was made in constructing series representations 
for multivariable sigma functions on Jacobian varieties of canonical plane algebraic curves, 
see \cite{bel1999, bl2002, bl2004, bl2008}  (or \cite{belMDSF} in whole).
Such a sigma function is defined by a system of heat equations. 
The system has a unique solution if a curve has a canonical form called
an $(n,s)$-curve with co-prime $n$ and $s$. Introduced in 
\cite{bel1999}, $(n,s)$-curves serve as a generalization of the Weierstrass canonical form of elliptic curves.

The method of constructing series of multivariable sigma functions
is illustrated in \cite{egoy}, where the systems of heat equations and the series 
for the $(2,5)$, $(2,7)$ and $(3,4)$ curves are constructed explicitly. 
Series expansions of the sigma functions for
the $(3,4)$ and $(3,5)$ curves can be found in \cite{bl2018}.

Sigma functions for space curves are constructed in \cite{mats,KMP1,KMP2}.

A uniformization of Jacobian varieties of trigonal $(n,s)$-curves,
which is, in fact, a solution of the Jacobi inversion problem, is proposed in \cite{bel00}.
The method based on the Klein formula \cite[Eq.\;(4.1)]{bel00}, which connects
the fundamental bi-differential to a quadratic function in differentials of the first kind
with $\wp$ functions as coefficients, was used.  Below we suggest an alternative method of 
solving the Jacobi inversion problem, which avoids computing the fundamental bi-differential.

Another approach to solving 
the Jacobi inversion problem uses ratios of theta functions
instead of multiply periodic $\wp$ functions. The two approaches are 
closely related, see  \cite[Art.\;212--216]{bakerAF}. So, 
ratios of theta functions can be obtained from the solution proposed in this paper.
An idea how to modify a theta function in the case of 
the extended Jacobi inversion problem\footnote{A solution of the Jacobi inversion problem 
is supposed to be a divisor of degree equal to the genus of a curve. If
a divisor of degree greater than the genus of a curve is required, then  
the problem of inverting Abel's map is called
the extended Jacobi inversion problem.}
is proposed in \cite{bf2008} as a generalization of  \cite{fed}.

Below we develop the ideas of \cite{bakerMPF} and \cite{belHKF},
and suggest a solution of  the Jacobi inversion problem on non-hyperelliptic curves
in terms of multiply periodic $\wp$ functions, also called Kleinian $\wp$ functions.
This solution provides a non-special divisor of the degree equal to the genus of a curve. 
The case of special divisors is a subject of another investigation. 
The extended Jacobi inversion problem can be solved by degenerating the sigma function of a
curve of a higher genus, see \cite{bl2019} for the case of finding a degree $2$ divisor on an elliptic curve.

The paper is organized as follows. In Preliminaries we briefly recall the notions of an $(n,s)$-curve
and the Sat\={o} weight, how to construct entire rational functions
as well as differentials of the first and the second kinds on such a curve,
and the known solutions of the Jacobi inversion problem.
In section 3 we give a solution of the Jacobi inversion problem on a non-hyperelliptic curve in general,
and display how to adapt the proposed method to hyperelliptic curves. In section 4
an explicit solution of the Jacobi inversion problem on trigonal curves is given. 
Similar solution on tetragonal and pentagonal curves are presented in sections 5 and 6, respectively.

\section{Preliminaries}
\subsection{$(n,s)$-Curves}
In \cite{bel1999}, with a pair of fixed co-prime integers $n$ and $s$, $n<s$,  a family of curves is defined:
\begin{subequations}\label{nsCurve}
\begin{equation}
\mathcal{V}_{(n,s)}=\{(x,y)\in \Complex^2 \mid f(x,y;\lambda) =0\},
\end{equation} 
where 
\begin{gather}
 f(x,y;\lambda) = -y^n + x^s + \sum_{j=0}^{n-2} \sum_{i=0}^{s-2}  \lambda_{ns-in- js} y^j x^i, \label{fEq}\\
\lambda_{k\leqslant 0}=0, \quad \lambda_{k}\in \Complex. \label{ModCond}
\end{gather}
\end{subequations}
The parameters $\lambda\equiv (\lambda_k)$ of such a family are varying.
The condition \eqref{ModCond} guarantees that
the genus of a curve from  $\mathcal{V}_{(n,s)}$ does not exceed
\begin{equation}\label{Vgenus}
g = \tfrac{1}{2} (n-1)(s-1). 
\end{equation}
Every family $\mathcal{V}_{(n,s)}$ with fixed $n$ and $s$ is considered as a fibre bundle over
the space of parameters $\Lambda = \{\lambda \in \Complex^{2g-M}\} \simeq \Complex^{2g-M}$. 
Here $M$ is called the modality, and
denotes the number of parameters $\lambda_{k\leqslant 0}$ assigned to zero. 

Let $ \mathcal{V}_{(n,s)}^0 \,{\subset}\, \mathcal{V}_{(n,s)}$ be the submanifold of
degenerate curves with genera less than~$g$.
In what follows, we always consider curves from 
$\mathcal{V}_{(n,s)}^c = \mathcal{V}_{(n,s)} \backslash \mathcal{V}_{(n,s)}^0$. 
That is all branch points of a curve 
are distinct, and the genus of a curve equals~$g$ defined by \eqref{Vgenus}. Any degenerate curve can be 
factorized into a non-degenerate curve, and a number of planes,
see for example \cite{bl2019}. The non-degenerate curve obtained from such a factorization belongs to
another family $\mathcal{V}$. Then the Jacobi inversion problem on $\mathcal{V}$
 can be considered instead. 

All $(n,s)$-curves have the property that infinity is a single point which is a branch point where all $n$ sheets join together.
This point serves as the base point. In the vicinity of infinity on the curve \eqref{nsCurve} 
the following expansions in a local parameter~$\xi$  hold
\begin{gather}\label{nsParam}
x = \xi^{-n},\qquad y = \xi^{-s}(1+O(\lambda)).
\end{gather}
The negative exponent of the leading term in the expansion about infinity serves as the Sat\={o} weight.
Thus, the Sat\={o} weights of $x$  and $y$ are $\wgt x = n$, $\wgt y=s$.
The weights are assigned to parameters of the curve:  $\wgt \lambda_k = k$.
Note, that $f$ is homogeneous with respect to the Sat\={o} weight, and $\wgt f= ns$.

The Sat\={o} weight introduces an order in the set of monomials $y^jx^i$, that is $\wgt y^jx^i = js+in$.
Note that, the Sat\={o} weight shows the order of the pole at infinity of a monomial.
We use the ordered list of monomials $\mathfrak{M}$ as a characteristic of an $(n,s)$-curve. Actually,
$\mathfrak{M} = \{\mathcal{M}_{ js+in - 2g+1} =  y^jx^i \mid i \geqslant 0,\, 0 \leqslant j \leqslant n-1\}$.

\subsection{Abel's map}
Let $\mathfrak{W}_{(n,s)} = \{\mathfrak{w}_1,\, \mathfrak{w}_2,\,\dots,\, \mathfrak{w}_g\}$ be the 
Weierstrass gap sequence of $\mathcal{V}^c_{(n,s)}$ of genus $g$, namely 
$$\mathfrak{W}_{(n,s)} = (\{0\} \cup \Natural) \backslash \{js+in \mid i,\, j \geqslant 0\}.$$
Note, that $(n,s)$-curves cover all plane algebraic curves with Weierstrass gap sequences
generated by two co-prime numbers. Every class of bi-rationally equivalent plane algebraic curves 
characterized by a fixed Weierstrass gap sequence $\mathfrak{W}_{(n,s)}$ has a representative 
$(n,s)$-curve. 

Let $\rmd u \equiv (\rmd u_{\mathfrak{w}_1},\,\rmd u_{\mathfrak{w}_2},\,\dots,\,\rmd u_{\mathfrak{w}_g})^t$ denote
differentials of the first kind on the curve. Actually,
\begin{gather}\label{uDif}
\rmd u_{\mathfrak{w}_i} = \frac{\mathcal{M}_{-\wFr_i} \,\rmd x}{\partial_y f(x,y)},\quad   1 \leqslant i \leqslant g,
\end{gather}
where $\mathfrak{w}_i$ runs the Weierstrass  gap sequence, and $\wgt \rmd u_{\mathfrak{w}_i} = - \mathfrak{w}_i$.
Note that the first $g$ monomials $\{\mathcal{M}_{-\mathfrak{w}_g}$, \ldots, $\mathcal{M}_{-\mathfrak{w}_2}$, 
$\mathcal{M}_{-\mathfrak{w}_1}\}$ from the list $\mathfrak{M}$ are employed.

Let $\rmd r \equiv (\rmd r_{\mathfrak{w}_1}$, $\rmd r_{\mathfrak{w}_2}$, \ldots, $\rmd r_{\mathfrak{w}_g})^t$ denote
differentials of the the second kind,
$\wgt \rmd r_{\mathfrak{w}} = \mathfrak{w}$. Let
\begin{gather}\label{rDif}
\rmd r_{\wFr_i} = \bigg(\wFr_i \mathcal{M}_{\wFr_i}  
+ \sum_{-2g < \kappa < \wFr_i} d_\kappa (\lambda) \mathcal{M}_\kappa \bigg) 
\frac{\rmd x}{\partial_y f(x,y)},\quad   i =1,\, \dots,\,g,
\end{gather}
and coefficients $d_\kappa$ are polynomials in $\lambda$.
The following condition completely determines the principle parts of $\rmd r$:
\begin{gather}\label{rCond}
\res_{\xi = 0} \bigg( \int_0^\xi \rmd u (\xi) \bigg) \, \rmd r(\xi)^t = 1_g,
\end{gather}
where $\rmd u (\xi)$ and  $\rmd r (\xi)$ denote expansions near infinity of differentials of the first and second kind,
and $1_g$ denotes the identity matrix of size $g$. The holomorphic parts of $\rmd r$ are not essential in the
further computations. In what follows, we deal with $\rmd r_\ell$ 
of weights $1\leqslant \ell \leqslant n-1$, and of the form \eqref{rDif} with
$0 < \kappa < \wFr_i$.

Let $P$ be a point of a fixed curve from  $\mathcal{V}^c_{(n,s)}$, then
\begin{gather*}
u(P) = \int_{\infty}^P \rmd u, \qquad\qquad    r(P) = \int_{\infty}^P \rmd r
\end{gather*}
are integrals of the first and second kinds on the curve, respectively.
In fact,  $u$ serves as Abel's map on the curve, that is $u(P) \equiv \mathcal{A}(P)$.
Abel's map of a divisor $D= \sum_{i=1}^n P_i$ is defined by  $\mathcal{A}(D) = \sum_{i=1}^n \mathcal{A}(P_i)$. 
The definition of $r$ requires a regularization since $\rmd r$ 
have a singularity at infinity, for more details see \cite{bl2018}. The integrals $r$ of the second kind 
define zeta functions $\zeta$ on the curve, up to adding abelian functions.

Let $\{\mathfrak{a}_n, \mathfrak{b}_n \mid n =1,\,\dots,\, g\}$ be canonical cycles on the curve,
that is they form a symplectic basis of the first homology group $\mathcal{H}_1(\mathcal{V}^c_{(n,s)})$.  Let 
$\omega = (\omega_{in})$ $\omega' = (\omega'_{in})$ denote period matrices of the first kind:
\begin{gather*}
\omega_{in} = \int_{\mathfrak{a}_n} \rmd u_{\mathfrak{w}_i}, \qquad\qquad    
\omega'_{in} = \int_{\mathfrak{b}_n} \rmd u_{\mathfrak{w}_i}.
\end{gather*}
If $\mathfrak{P}$ denotes the lattice of periods generated by the columns of $(\omega,\omega')$,
then $\Jac = \Complex^g / \mathfrak{P}$ 
is a Jacobian variety (Jacobian) of the curve under consideration. 
We denote coordinates of $\Jac$, and of the space $\Complex^g$ where $\Jac$ is embedded,
 by $u=(u_{\mathfrak{w}_1},\,u_{\mathfrak{w}_2},\, 
\dots,\, u_{\mathfrak{w}_g})^t$, $\wgt u_{\mathfrak{w}}=-\mathfrak{w}$.

\subsection{Sigma function}
Let $\sigma_{(n,s)}(u;\lambda)$ be an entire function called the sigma function 
of a family $\mathcal{V}^c_{(n,s)}$ of curves. We employ the definition given in \cite{bl2004}.
Let $\Lambda$ be the space of parameters of $\mathcal{V}^c_{(n,s)}$, 
and  $\mathcal{U}$ denote the space of Jacobians over $\Lambda$, called 
 the universal space.

\begin{defin} 
A multivariable sigma function $\sigma_{(n,s)}: \Complex^g\times \Lambda \mapsto \Complex$ 
is defined by a system of heat equations $Q_k \sigma_{(n,s)} =0$ and the initial condition given by
the Weierstrass-Schur polynomial corresponding to the Weierstrass gap sequence of $\mathcal{V}^c_{(n,s)}$.
 The annihilation operators $Q_k$, $k=1$, \ldots, $\dim \Lambda$, 
 form a Lie algebra on the universal space $\mathcal{U}$ of $\mathcal{V}^c_{(n,s)}$, obtained as a lift of the Lie algebra of 
 vector fields $L_k$ on~$\Lambda$. 
\end{defin}
Vector fields  $L_k$ are differential operators of the first order with respect to  $\lambda\in \Lambda$.
They serve as a frame on the first de Ram cohomology group $\mathcal{H}^1(\mathcal{V}_{(n,s)})$,  and
produce the Gauss-Manin connection $\Gamma_k$: 
$$L_k R = \Gamma_k R + \rmd \varphi_k(x,y,\lambda),\quad R \in \mathcal{H}^1.$$
By a linear transformation cohomologies $R$ are made symplectic, that is 
the corresponding period matrix $\Omega$ satisfies the Legendre relation: $\Omega^t J \Omega = 2\pi \imath J$.
The corresponding gauge transformation of the Gauss-Manin connection leads 
to $\widetilde{\Gamma}_k \in \mathfrak{sp}(2g)$, which produces
differential operators $H_k$ of the second order  with respect to $u\in \Complex^g$:
\begin{equation*}
H_k = 
\left\langle \widetilde{\Gamma}_k J \small \begin{pmatrix} \partial_{u} \\ u \end{pmatrix},
\begin{pmatrix} \partial_{u} \\ u \end{pmatrix} \right\rangle.
\end{equation*}
Then the annihilation operators $Q_k$ are obtained by
\begin{equation}\label{AnnihOp}
Q_k = \frac{1}{2} H_{k} + L_k +  \frac{1}{8}  L_k \big(\log \det V(\lambda)\big),
\end{equation}
where $V(\lambda)$ is the coefficient matrix of the vector fields $L_k$, that is
 $L=V(\lambda) \partial_\lambda$ in the matrix form.

A system of heat equations with the annihilation operators \eqref{AnnihOp}
produces a series  representation of $\sigma_{(n,s)}$. 
Such a technique is illustrated in more detail in \cite{egoy}.
The series of $\sigma_{(n,s)}$ is analytic
in $u$ and $\lambda$, and (see \cite{bel1999})
 $$\wgt \sigma_{(n,s)} = - \tfrac{1}{24} (n^2-1)(s^2-1).$$

Zeta functions and abelian functions on the Jacobian of $\mathcal{V}_{(n,s)}$
 are defined with the help of its multivariable sigma function $\sigma$:
 \begin{subequations}\label{wpDefs}
\begin{align}
 &\zeta_{i}(u) =  \frac{\partial}{\partial u_i} \log \sigma(u), \label{zetaDef} \\
 &\wp_{i,j}(u) = - \frac{\partial^2}{\partial u_i \partial u_j} \log \sigma(u), \label{wp2Def}\\
  &\wp_{i,j,k}(u) = - \frac{\partial^3}{\partial u_i \partial u_j \partial u_k} \log \sigma(u), \label{wp3Def} \quad \text{etc.}
\end{align}
\end{subequations}
For brevity, we write $\sigma(u)$, $\zeta_i(u)$, $\wp_{i,j}(u)$, etc. instead 
of $\sigma(u;\lambda)$,  $\zeta_i(u;\lambda)$, $\wp_{i,j}(u;\lambda)$, etc.
The abelian functions are periodic over the lattice $\mathfrak{P}$ defined above. 

\subsection{The Jacobi inversion problem}\label{ss:JIP}
Let $\mathcal{S}^g$ be a symmetric product of $g$ copies of the same curve $V$ of genus $g$,
and $\mathcal{D}_g \subset \mathcal{S}^g$ consists of all degree $g$ non-special divisors.
Abel's map establishes a one-to-one correspondence between $\mathcal{D}_g$ and 
$\Jac_g \equiv \mathcal{A}(\mathcal{D}_g) \subset \Jac$.
Note that $\sigma(u)\neq 0$ if $u \in \Jac_g$, and
$\sigma(u)=0$ if $u \in \Jac_0 \equiv \Jac \backslash \Jac_g$
by the Riemann vanishing theorem.

\newtheorem*{JIProblem}{The Jacobi inversion problem}
\begin{JIProblem}
Given $u \in \Jac_g$ such that $\sigma(u)$  does not vanish,
find a non-special degree $g$ divisor $D_g$ such that $\mathcal{A}(D_g) = u$.
\end{JIProblem}

The Jacobi inversion problem in this formulation has a unique solution. 
Every divisor $D_g \in \mathcal{D}_g$ serves as a representative of a class of equivalent divisors
$\{D \mid \deg D \geqslant g,\, \mathcal{A}(D) = \mathcal{A}(D_g)\}$. There is only one divisor $D_g$ of degree $g$
in each class, and it is called a reduced divisor. We leave the Jacobi inversion problem for points from $\Jac_0 $ 
beyond our consideration in this paper.

Now we recall the known solutions of the declared Jacobi inversion problem.

\begin{exam}
A uniformization of the Weierstrass canonical curve ${-}y^2 + 4x^3 - g_2 x - g_3 =0$
is given by $(x,y)= \big(\wp(u;g_2,g_3), - \wp' (u;g_2,g_3)\big)$ with the standard Weierstrass
$\wp$-function and $\sigma$-function.

The Weierstrass canonical curve is equivalent to $\mathcal{V}_{(2,3)}$ of the form 
$-y^2 + x^3 + \lambda_4 x + \lambda_6 = 0$, and the corresponding sigma function $\sigma_{(2,3)}$ relates
 to the Weierstrass $\sigma$-function as
$\sigma_{(2,3)}(u; \lambda_4,\lambda_6) = \sigma (u; -4 \lambda_4, -4 \lambda_6)$. 
Then $\wp_{(2,3)}(u; \lambda_4,\lambda_6) = 2 \wp(u; -4 \lambda_4, -4 \lambda_6)$,
and a uniformization of $\mathcal{V}_{(2,3)}$ is given by $(x,y)= \big(\wp_{(2,3)}(u; \lambda)$, 
$- \tfrac{1}{2} \wp'_{(2,3)} (u; \lambda)\big)$.
\end{exam}

\begin{exam}
In \cite[Art.\;11]{bakerMPF} the reader can find a solution 
of the Jacobi inversion problem on $(2,5)$-curve. 
If $D=(x_1,\,y_1)+(x_2,\,y_2)$, and $u=\mathcal{A}(D)$, then\footnote{The given formulas 
are obtained with the following differentials of the first and second kinds 
\begin{align*} 
& \rmd u_1 = \frac{x\,\rmd x}{-2y},& &\rmd r_3 = \frac{x^2 \rmd x}{-2y},&\\
& \rmd u_3 = \frac{\rmd x}{-2y},& &\rmd r_1 = (3x^3 +\lambda_4 x) \frac{\rmd x}{-2y}.&
\end{align*}}
\begin{subequations}\label{JIP25}
\begin{gather}
x_1+x_2 = \wp_{1,1}(u), \qquad   x_1 x_2 = - \wp_{1,3}(u),\\
y_i = -\tfrac{1}{2} \big( x_i \wp_{1,1,1}(u) +  \wp_{1,1,3}(u) \big),\qquad i=1,\, 2.
\end{gather}
\end{subequations}
\end{exam}

\begin{exam}[\textbf{Hyperelliptic curves}]
A solution of the Jacobi inversion problem on a hyperelliptic curve is proposed in  \cite[Art.\;216]{bakerAF}
and rediscovered in \cite[Theorem 2.2]{belHKF}.
Let a non-degenerate hyperelliptic curve of genus $g$ be defined\footnote{A $(2,2g+1)$-curve serves as a 
canonical form of hyperelliptic curves of genus $g$.} by
\begin{equation}\label{V22g1Eq}
-y^2 + x^{2 g+1} + \sum_{i=1}^{2g} \lambda_{2i+2} x^{2g-i} = 0.
\end{equation}
Let $u = \mathcal{A}(D)$ be the Abel image of  a degree $g$ non-special  divisor  $D$
on the curve. Then $D$ is uniquely defined by the system of equations 
\begin{subequations}\label{EnC22g1}
\begin{align}
&x^{g} -  \sum_{i=1}^{g} x^{g-i}  \wp_{1,2i-1}(u) = 0,\\ 
&2 y + \sum_{i=1}^{g} x^{g-i}  \wp_{1,1,2i-1}(u) = 0.
\end{align}
\end{subequations}
\end{exam}

\subsection{Entire rational functions on a curve}\label{ss:EnRatF} 
Let $\mathcal{R}$ be an entire rational function on a curve $V$ of genus $g$
from $\mathcal{V}^c_{(n,s)}$. That is, $\mathcal{R}$ is analytic on the curve
punctured at infinity,  and has a pole at infinity.
Let $\mathcal{R}$ be of weight $N$, that is the divisor of zeros consists of $N$ points,
and a pole of order $N$ is located at infinity. The function $\mathcal{R}$ of weight $N\geqslant 2g$
 has the form of a linear combination 
of the first $N-g+1$ monomials from the ordered list $\mathfrak{M}$. Indeed,
according to the Riemann---Roch theorem, only $N-g$ zeros can be chosen arbitrary, the remaining $g$ zeros
are determined by the curve equation. In general, an entire rational function has the form
\begin{equation}\label{RN}
\mathcal{R}(x,y) = \sum_{j=0}^{n-1} y^j \rho_{j}(x),
\end{equation}
where $\rho_{j}$, $0\leqslant j \leqslant n-1$, are polynomials in $x$.

\section{Solving the Jacobi inversion problem}
In what follows, curves are supposed to be non-hyperelliptic. The hyperelliptic case is considered separately in 
Example~\ref{E:HypC}.
\begin{theo}\label{T1}
Let $D$ be a positive non-special  divisor of degree $g$ 
on a non-degenerate $(n,s)$-curve of genus $g$. 
Then there exist $n-1$ entire rational functions $\mathcal{R}_{2g+l}$ of the weights $2g+l$, $l=0$, \ldots, $n-2$, 
which vanish on $D$, and the following system of equations defines $D$ uniquely 
\begin{gather}\label{REqs}
 \begin{split}
 &\mathcal{R}_{2g}(x,y) = 0,\\ 
 &\mathcal{R}_{2g+1}(x,y) = 0,\\
  &\vdots\\
 &\mathcal{R}_{2g+n-2}(x,y) = 0.
 \end{split}
 \end{gather}
\end{theo}

\begin{theo}\label{T2}
Let $V$ be a non-degenerate $(n,s)$-curve of genus $g$
with a list of monomials $\mathfrak{M}$  sorted in the ascending order of the  Sat\={o} weight, 
and $\mathcal{M}_{\wFr}$ denotes a monomial of weight $\wFr+2g-1$.
Let $\widetilde{\mathcal{M}}_{\ell}$ be the entire rational function 
equal to $(\partial_y f)\rmd r_\ell/ \rmd x$, and the differentials of the second kind $\rmd r_\ell$, $\ell=1$,
\ldots, $n-1$, be defined by  \eqref{rCond}.
Then the entire rational functions in \eqref{REqs}  have the form
\begin{equation}\label{R2giAb}
\mathcal{R}_{2g+\ell-1}(x,y;u) =  \widetilde{\mathcal{M}}_{\ell} 
- \sum_{i=1}^{g} A_{\ell,\wFr_i}(u) \mathcal{M}_{-\wFr_i}, \qquad  
1 \leqslant \ell \leqslant n-1,
\end{equation}
where $A_{\ell,\wFr_i}$ denote abelian functions on the Jacobian variety of $V$,
and $\wFr_i$ runs the Weierstrass gap sequence.
\end{theo}

\begin{proof}[Proof of Theorem~\ref{T1}]
The entire rational functions in \eqref{REqs} have the form
 \begin{equation}\label{R2gi}
\mathcal{R}_{2g+l}(x,y) = \sum_{j=0}^{n-2} y^j \rho_{j}^{[2g+l]}(x), \qquad  0 \leqslant l \leqslant n-2,
\end{equation}
where $\rho_{j}^{[2g+l]}$ denotes a polynomial in $x$. Note, that 
entire rational functions of the weights not greater than $2g+n-2$ do not contain the term $y^{n-1} \rho_{n-1}(x)$. 
That is, the highest degree of $y$ in \eqref{R2gi} is $n-2$.
Indeed, since $\wgt y^{n-1} = s(n-1)$ is greater than $2g+n-2=s (n-1)-1$, the minimal weight of a function which contains
the term $y^{n-1}$ is $2g+n-1$.

Suppose, we are given a positive non-special  divisor $D$  of degree $g$, namely 
$D=\sum_{k=1}^g (x_k,y_k)$. One can construct an entire rational function of weight $2g+l$, $l\geqslant 0$, 
 vanishing on $D$ by a determinant composed from 
 the first $g + l +1$ monomials of the ordered list $\mathfrak{M}$ 
corresponding to a curve under consideration:
\begin{multline*}
\mathcal{R}_{2g+l}(x,y) =\\
\begin{vmatrix}
\mathcal{M}_{-\mathfrak{w}_g}(x,y) & \ldots &
\mathcal{M}_{-\mathfrak{w}_1}(x,y) & \mathcal{M}_{1}(x,y) & \ldots & \mathcal{M}_{l+1}(x,y)  \\
\{\mathcal{M}_{-\mathfrak{w}_g}(x_k,y_k) & \ldots &
\mathcal{M}_{-\mathfrak{w}_1}(x_k,y_k) & \mathcal{M}_{1}(x_k,y_k) 
& \ldots & \mathcal{M}_{l+1}(x_k,y_k)\}_{k=1}^{g+l}
\end{vmatrix},
\end{multline*}
where the $(k+1)$-th row contains values of the monomials at $(x_k,y_k)$,
and the points $\{(x_k,y_k)\}_{k=1}^g$ $g$ form the support of $D$. 
Note, that the function $\mathcal{R}_{2g}$ is defined by $D$ uniquely up to a constant multiple.
The other functions $\mathcal{R}_{2g+l}$, $l>0$, require $l$ additional points to be defined uniquely. 
As a result, all functions $\mathcal{R}_{2g+l}$, $l\geqslant 0$, vanish on $D$.

The system \eqref{REqs} admits a matrix form
\begin{gather}\label{RY}
\tens{R}(x) \tens{Y} = 0,
\end{gather}
where
\begin{gather*}
 \tens{R}(x) = \begin{pmatrix} 
 \rho_{0}^{[2g]} & \rho_{1}^{[2g]} & \dots & \rho_{n-2}^{[2g]} \\ 
 \rho_{0}^{[2g+1]} & \rho_{1}^{[2g+1]} & \dots & \rho_{n-2}^{[2g+1]} \\ 
 \vdots & \vdots & \ddots & \vdots \\ 
 \rho_{0}^{[2g+n-2]} & \rho_{1}^{[2g+n-2]} & \dots & \rho_{n-2}^{[2g+n-2]}
 \end{pmatrix},\quad 
 \tens{Y} = \begin{pmatrix} 1 \\ y \\ \vdots \\ y^{n-2} \end{pmatrix}.
\end{gather*}
For brevity, we omit the arguments of polynomials
and write $\rho_{j}^{[2g+l]}$ instead of $\rho_{j}^{[2g+l]}(x)$.
Denote $\mathcal{X}(x)=\det \tens{R}(x)$.
Since all functions $\mathcal{R}_{2g+l}$, $l\geqslant 0$, vanish on $D$,
we have  $\deg \mathcal{X}(x) \geqslant g$.

Let $\mFr$ and $\rFr$ be the natural numbers  such that $s= n \mFr + \rFr$. That is $\mFr$
is the integer part of $s/n$ and $\rFr$ is the remainder.
Then
\begin{gather*}
\deg \rho_{j}^{[2g+l]} = \Big[ \frac{2g + l - j s }{n} \Big] 
=  \Big[\mFr (n-j-1) + \frac{\rFr}{n} (n-j-1) - \frac{ n - l - 1}{n} \Big].
\end{gather*}
Now, we compute the degree of $\mathcal{X}$ from its matrix form.
Actually, $\mathcal{X}$ is a sum of terms of the form $\prod_{j=0}^{n-2} \rho_{j}^{[2g+l(j)]}$,
where the function $l(j)$ is a permutation of the values $j$ running from $0$ to  $n-2$.
Then
\begin{multline*}
\deg \prod_{j=0}^{n-2} \rho_{j}^{[2g+l(j)]}
 \leqslant \Big[ \sum_{j=0}^{n-2}  \Big( \mFr (n-j-1) + \frac{\rFr}{n} (n-j-1) - \frac{1}{n}(n - l(j) - 1)\Big) \Big] \\
 = \Big[ \frac{1}{2} (n-1) n \mFr + \frac{\rFr}{2}(n-1) - \frac{1}{2} (n-1) \Big] 
 = \Big[ \frac{1}{2} (n-1) (s-1)\Big] = g.
\end{multline*}
Thus, $\deg \mathcal{X}(x) \leqslant g$.

Therefore, $\deg \mathcal{X}(x) = g$, and the divisor of zeros of $\mathcal{X}$ 
coincides with $x$-coordinates of the support of $D$. 
The corresponding $y$-coordinates are determined from the system \eqref{REqs}.
Repeated values of $x$ are allowed, unless $D$ contains $n$ points connected by 
involution\footnote{We say that $n$ points $(a,b_k)$, $k=1$, \ldots, $n$, are
connected by involution of an $(n,s)$-curve $f(x,y)=0$ if $b_k$ give all solutions
of $f(a,y)$=0.},
that turns $D$ into a special divisor. 
\end{proof}

\begin{proof}[Proof of Theorem~\ref{T2}]
We suppose that $u=\mathcal{A}(D)$, where $D$ is a degree $g$ non-special  divisor on $V$.
From the vanishing properties of the sigma function, we know that $\sigma\big(u - \mathcal{A}(x,y) \big)$
has zeros exactly at $D=\sum_{k=1}^g (x_k,\,y_k)$. By the residue theorem, we have
\begin{equation*}
 \sum_{k=1}^g r_{\ell} (x_k,\,y_k)  = 
 \frac{1}{2\pi \rmi}\oint_{C}  \bigg(\int_{\infty}^{(x,y)} \rmd r_{\ell} \bigg) 
 \, \rmd \log \sigma\big(u - \mathcal{A}(x,y) \big),
\end{equation*}
where $\rmd r_{\ell}$ is a differential, and $r_{\ell}$ is the corresponding integral of the second kind of weight $\ell$, 
and $C$ denotes a contour on the curve enclosing the divisor $D$, or equivalently, a contour cutting off
the infinity. After substituting the parameterization~\eqref{nsParam},
the integral on the right hand side transforms into one with respect to the parameter $\xi$
on $\mathbb{CP}^1$. The contour $C$ turns into one 
encircling $\xi=0$ in the clockwise direction. Note that $r_\ell$ has a pole at infinity ($\xi=0$), and so
we compute the integral through the residue at $\xi=0$, namely
\begin{equation}\label{rExpr}
 \sum_{k=1}^g r_{\ell} (x_k,\,y_k)  = - \res\limits_{\xi=0} \bigg(r_{\ell}(\xi) \frac{\rmd}{\rmd \xi}
 \log \sigma\big(u- \mathcal{A}(\xi)\big) \bigg) \equiv R_{\ell} (u).
\end{equation}
In fact, \eqref{rExpr} gives a definition of  $\zeta_{\ell}$.

Differentiating \eqref{rExpr} with respect to $x_1$, we find
\begin{equation}\label{RatFExpr}
\frac{\rmd r_{\ell}(x_1,\,y_1)}{\rmd x_1}  = \big(\partial_u R_{\ell} (u)  \big)^t\frac{\rmd u(x_1,\,y_1)}{\rmd x_1}.
\end{equation}
Multiplying \eqref{RatFExpr} by $\partial_{y_1} f(x_1,y_1)$, we find the relations \eqref{REqs}
with the entire rational functions of the form \eqref{R2giAb},
where $A_{\ell,\wFr_i}(u) = \partial_{u_{\wFr_i}} R_{\ell} (u)$. 
Note that $\rmd u_{\wFr}(x,y) / \rmd x= \mathcal{M}_{-\wFr} /(\partial_y f)$, where $\wFr$ runs 
the Weierstrass gap sequence  $\mathfrak{W}=\{\wFr_i\} $.  

Instead of the second kind differentials defined by  \eqref{rCond},
one can use a simpler form: $\rmd r_{\ell} = \ell \mathcal{M}_{\ell} \, \rmd x/(\partial_y f)$,
$\ell=1$, \ldots, $n-1$. Then by simplifying \eqref{rExpr} to the form with no $\zeta_k$, $k<\ell$,
one finds the relations \eqref{REqs}
with the entire rational functions of the form \eqref{R2giAb}.
\end{proof}

The proposed method is applicable to hyperelliptic curves. Though $n-1=1$ in this case,
two entire rational functions are needed to define  $D$ such that $u=\mathcal{A}(D)$.
\begin{exam}[\textbf{Hyperelliptic curves}]\label{E:HypC}
Recall that the curve \eqref{V22g1Eq} has
the Weierstrass sequence $\mathfrak{W}=\{2i-1 \mid i=1,\, \dots,\, g\}$. 
The list of monomials is $\mathfrak{M} = \{ y^j x^i \mid j=0,1,\, i =0,\,1,\, \dots \}$. The first $g$ monomials 
$\mathcal{M}_{-(2i-1)} = x^{g-i}$, $i=1,\, \dots,\, g$, produce differentials of the first kind.
 Then, we need differentials of the second kind of the weights 
$1$ and $2$. The latter are generated by the monomials $x^g$ and $y$, namely
\begin{align*}
 &(\rmd r_1(x,y),\, \rmd r_2(x,y)) = \frac{\rmd x}{\partial_y f} 
 \big( x^g,\, 2 y \big). 
\end{align*}
Using the method described in the proof of Theorem~\ref{T2},
we find that the preimage $D$ of $u=\mathcal{A}(D)$ is uniquely defined by the system
\begin{equation*}
\mathcal{R}_{2g}(x,y;u)=0,\qquad  \mathcal{R}_{2g+1}(x,y;u)=0
\end{equation*}
with two entire rational functions of the weights $2g$ and $2g+1$:
\begin{align*}
\mathcal{R}_{2g}(x,y;u) 
 &= x^{g} -  \sum_{i=1}^{g} \wp_{1,2i-1}(u) x^{g-i},\\ 
\mathcal{R}_{2g+1}(x,y;u) 
 &= 2 y + \sum_{i=1}^{g} \wp_{1,1,2i-1}(u) x^{g-i},
\end{align*}
which coincides with \eqref{EnC22g1}.
\end{exam}

Below we illustrate the proposed method with multiple examples. 
We consider trigonal, tetragonal and pentagonal curves.

\section{Jacobi inversion problem on trigonal curves}
Trigonal $(n,s)$-curves  split into two
types: $(3,3\mFr+1)$ and $(3,3\mFr+2)$-curves,
where $\mFr$ is a natural number.

\begin{theo}[$(3,3\mFr+1)$-Curves]\label{T:C33m1}
Let a non-degenerate $(3,3\mFr+1)$-curve of genus $g=3\mFr$, $\mFr \in \Natural$, be defined by
\begin{equation}\label{V33m1Eq}
-y^3 + x^{3\mFr+1} + y \sum_{i=0}^{2\mFr} \lambda_{3i+2} x^{2\mFr-i} + \sum_{i=1}^{3\mFr} \lambda_{3i+3} x^{3\mFr-i} = 0.
\end{equation}
Let $u = \mathcal{A}(D)$ be the Abel image of  a degree $g$ non-special  divisor  $D$
on the curve. Then $D$ is uniquely defined by the system of equations 
\begin{subequations}\label{REqsC33m1}
\begin{equation}
\mathcal{R}_{6\mFr}(x,y;u)=0,\qquad  \mathcal{R}_{6\mFr+1}(x,y;u)=0
\end{equation}
with two entire rational functions of the weights $2g=6\mFr$, $2g+1=6\mFr+1$:
\begin{align}
\mathcal{R}_{6\mFr}(x,y;u) 
 &= x^{2\mFr} -  \sum_{i=1}^{3\mFr} \wp_{1,\wFr_i}(u) \M_{-\wFr_i},\\ 
\mathcal{R}_{6\mFr+1}(x,y;u) 
 &= 2 y x^{\mFr} - \sum_{i=1}^{3\mFr} \big(\wp_{2,\wFr_i}(u) - \wp_{1,1,\wFr_i}(u) \big) \M_{-\wFr_i},
\end{align}
\end{subequations}
where
\begin{align*}
\M_{-(3i-2)} & = y x^{\mFr-i},\quad  i=1,\, \dots,\, \mFr,\\
\M_{-(3i-1)} & = x^{2\mFr-i}, \quad  i=1,\, \dots,\, 2\mFr.
\end{align*}
\end{theo}

\begin{theo}[$(3,3\mFr+2)$-Curves]\label{T:C33m2}
Let a non-degenerate $(3,3\mFr+2)$-curve of genus $g=3\mFr+1$, $\mFr \in \Natural$, be defined by
\begin{equation}\label{V33m2Eq}
-y^3 + x^{3\mFr+2} + y \sum_{i=0}^{2\mFr + 1} \lambda_{3i+1} x^{2\mFr+1-i} 
+ \sum_{i=1}^{3\mFr +1} \lambda_{3i+3} x^{3\mFr+1-i} = 0.
\end{equation}
Let $u = \mathcal{A}(D)$ be the Abel image of  a degree $g$ non-special divisor  $D$
on the curve. Then $D$ is uniquely defined by the system of equations 
\begin{subequations}\label{REqsC33m2}
\begin{equation}
\mathcal{R}_{6\mFr+2}(x,y;u)=0,\qquad  \mathcal{R}_{6\mFr+3}(x,y;u)=0
\end{equation}
with two entire rational functions of the weights $2g=6\mFr+2$, $2g+1=6\mFr+3$:
\begin{align}
\mathcal{R}_{6\mFr+2}(x,y;u) 
 &= y x^{\mFr} - \sum_{i=1}^{3\mFr+1} \wp_{1,\wFr_i}(u) \M_{-\wFr_i},\\ 
\mathcal{R}_{6\mFr+3}(x,y;u) 
 &= 2 x^{2\mFr+1} + \lambda_1 y x^{\mFr}
 - \sum_{i=1}^{3\mFr+1}  \big(\wp_{2,\wFr_i}(u) - \wp_{1,1,\wFr_i}(u) \big)  \M_{-\wFr_i} ,
\end{align}
\end{subequations}
where
\begin{align*}
\M_{-(3i-1)} & = y x^{\mFr-i},\quad  i=1,\, \dots,\, \mFr,\\
\M_{-(3i-2)} & = x^{2\mFr+1-i}, \quad  i=1,\, \dots,\, 2\mFr +1.
\end{align*}
\end{theo}

\subsection{Proof of Theorem~\ref{T:C33m1} ($(3,3\mFr+1)$-Curves)}
Let $\xi$ be a local parameter in the vicinity of infinity on the curve \eqref{V33m1Eq}, 
then
\begin{equation}\label{xyParamC33m1}
 x(\xi) = \xi^{-3},\qquad y(\xi) = \xi^{-3\mFr-1}\Big(1 + \frac{\lambda_2}{3} \xi^2 + 
 \frac{\lambda_5}{3}  \xi^5 + O(\xi^6)\Big).
\end{equation}
The basis of differentials of the first kind has the form
\begin{align*}
 \rmd u(x,y) &= \frac{\rmd x}{\partial_y f}  \big(y x^{\mFr-1},\, x^{2\mFr-1},\, \dots,\, y,\, x^{\mFr},\, \dots,\,x,\,1 \big)^t\\
 &= \frac{\rmd x}{\partial_y f}  \big(\M_{-1},\, \M_{-2},\, \dots,\, \M_{-(3\mFr-2)},\,  \M_{-(3\mFr-1)}, \notag\\
 &\phantom{mmmmmmmmmmn} \dots,\, \M_{-(6\mFr-4)},\, \M_{-(6\mFr-1)} \big)^t, \notag
\end{align*}
and the expansions near infinity are
\begin{align*}
&\rmd u_{3i-2} = y x^{\mFr-i}\frac{ \rmd x}{\partial_y f} 
=\xi^{3i-3} \Big(1 + O(\xi^4)\Big) \rmd \xi,\quad i = 1,\, \dots,\, \mFr,\\
&\rmd u_{3i-1} = x^{2\mFr-i}\frac{ \rmd x}{\partial_y f}  
= \xi^{3i-2} \Big(1 + O(\xi^2)\Big) \rmd \xi,\quad i = 1,\, \dots,\, 2\mFr.
\end{align*}
By integration with respect to $\xi$ we find expansions for integrals of the first kind:
\begin{subequations}\label{Int1s33m1} 
\begin{align}
& u_{3i-2}(\xi) = \frac{\xi^{3i-2}}{3i-2} 
+ O(\xi^{3i+2}),\ i = 1,\, \dots,\, \mFr,\\
& u_{3i-1}(\xi) = \frac{\xi^{3i-1}}{3i-1} 
+ O(\xi^{3i+1}),\ i = 1,\, \dots,\, 2\mFr.
\end{align}
\end{subequations}

The two differentials of the second kind of the smallest Sat\={o} weights are
\begin{align}\label{Dif2s33m1}
 &(\rmd r_1(x,y),\, \rmd r_2(x,y)) = \frac{\rmd x}{\partial_y f} 
 \big( x^{2\mFr},\, 2yx^{\mFr} \big). 
\end{align}
With the help of \eqref{xyParamC33m1} we find the corresponding expansions near infinity:
\begin{align*}
&\rmd r_1(\xi) = \xi^{-2} \Big(1 + O(\xi^2)\Big) \rmd \xi,\\
&\rmd r_2(\xi) = 2 \xi^{-3} \Big(1  + O(\xi^4)\Big) \rmd \xi.
\end{align*}
As seen from a direct computation, \eqref{Dif2s33m1} satisfy the condition \eqref{rCond}.
Next, we find the integrals of the second kind
\begin{align*}
& r_1(\xi) = -\xi^{-1}  + O(\xi) + c_1,\\
& r_2(\xi) = - \xi^{-2} + O(\xi^2) + c_2,
\end{align*}
where $c_1$ and $c_2$ are regularization constants\footnote{The reader can find a solution of
the regularization problem for integrals of the second kind on non-hyperelliptic curves in \cite{bl2018}, in particular,
$c_1 = 0$,  $c_2 = -\lambda_2/3$ on $\mathcal{V}_{(3,4)}$, and $c_1 = 0$,  $c_2 = -2\lambda_2/3$ on $\mathcal{V}_{(3,7)}$.}, which are not essential in the further computations.

The integrals of the first kind \eqref{Int1s33m1} define $\mathcal{A}(\xi)$ from \eqref{rExpr}, namely
$$\mathcal{A}(\xi) = \big(u_1(\xi),\, u_2(\xi),\, \dots,\,  u_{3\mFr-2}(\xi),\, u_{3\mFr-1}(\xi),\, \dots,\,  
u_{6\mFr-4}(\xi),\, u_{6\mFr-1}(\xi)\big)^t.$$
Then we compute residues on the right hand side of \eqref{rExpr} from the series expansions in~$\xi$.
Taking into account that $\mathcal{A}(0)=0$, $\mathcal{A}'(0)=(\delta_{i,1})$, $\mathcal{A}''(0)=(\delta_{i,2})$, 
where $\delta_{i,k}$ denotes the Kronecker delta and $i$ runs from $1$ to $g=3\mFr$, we find
\begin{equation*}
\frac{\rmd}{\rmd \xi} \log \sigma\big(u - \mathcal{A}(\xi)\big) = - \zeta_1(u)
- \big(\zeta_2(u) + \wp_{1,1}(u)\big) \xi + O(\xi^2).
\end{equation*}
Here $u = \mathcal{A}(D)$ is the Abel image of a non-special divisor 
$D = \sum_{k=1}^{3\mFr} (x_k,y_k)$.
Thus, \eqref{rExpr} acquires the form
\begin{gather}\label{ZetaC33m1}
\sum_{k=1}^{3\mFr} \begin{pmatrix}
r_1(x_k,\,y_k) \\ r_2(x_k,\,y_k) 
\end{pmatrix} = - \begin{pmatrix} \zeta_1(u) \\
\zeta_2(u) + \wp_{1,1}(u) \end{pmatrix}
\equiv \begin{pmatrix} R_1 (u) \\  R_2 (u)\end{pmatrix}.
\end{gather}
Differentiating the above relations with respect to $x_1$,  we obtain the relations
\begin{align*}
x_1^{2\mFr}
&= \bigg( y_1 \sum_{i=1}^{\mFr} \wp_{1,3i-2}(u) x_1^{\mFr-i} 
+ \sum_{i=1}^{2\mFr} \wp_{1,3i-1}(u) x_1^{2\mFr-i} \bigg),\\
2y_1 x_1^{\mFr}  &= 
\bigg( y_1 \sum_{i=1}^{\mFr} \big(\wp_{2,3i-2}(u) - \wp_{1,1,3i-2}(u) \big) x_1^{\mFr-i} \\
&\qquad + \sum_{i=1}^{2\mFr} \big( \wp_{2,3i-1}(u) - \wp_{1,1,3i-1}(u) \big) x_1^{2\mFr-i} \bigg),
\end{align*} 
which coincide with
\eqref{REqsC33m1}. $\qede$

\subsection{Example: $(3,4)$-Curve}
The family $\mathcal{V}_{(3,4)}$ of genus $3$ is defined by the equation
\begin{equation*}
-y^3 + x^4 + y (\lambda_2 x^2 + \lambda_5 x + \lambda_8)
+ \lambda_6 x^2 + \lambda_9 x + \lambda_{12} = 0,
\end{equation*}
with the Weierstrass gap sequence $\{1,\, 2, \,5\}$.
The first three monomials in the list $\mathfrak{M}$, sorted in the ascending order of the Sat\={o} weight, are
$\M_{-1} = y$, $\M_{-2} = x$,  $\M_{-5} = 1$. They serve as numerators of differentials of the first kind.
The next two monomials are $\M_{1}= x^2$,  $\M_{2}=y x$; they produce differentials of the second kind
$\rmd r_1$  and $\rmd r_2$ of the form \eqref{Dif2s33m1}.

A solution of the Jacobi inversion problem for $D=\sum_{k=1}^3 (x_k,y_k)$ on $\mathcal{V}_{(3,4)}$ 
such that $u = \mathcal{A}(D)$
 is given by the system
\begin{align*}
0=\mathcal{R}_6 (x,y;u) &\equiv x^2 - y \wp_{1,1}(u) - x \wp_{1,2}(u) - \wp_{1,5}(u), \\
0=\mathcal{R}_7 (x,y;u) &\equiv 2 y x - y \big(\wp_{1,2}(u) - \wp_{1,1,1}(u) \big) - 
x \big(\wp_{2,2}(u) - \wp_{1,1,2}(u) \big) \\
&\qquad \quad - \big(\wp_{2,5}(u) - \wp_{1,1,5}(u) \big).
\end{align*}

\begin{rem}
Let a trigonal curve of genus $3$ be defined by
the equation
\begin{equation*}
-y^3 + x^4 + y^2 (\lambda_1 x + \lambda_4)
+ y (\lambda_2 x^2 + \lambda_5 x + \lambda_8)
+ \lambda_3 x^3 + \lambda_6 x^2 + \lambda_9 x + \lambda_{12} = 0,
\end{equation*}
with the extra terms $y^2$, $x^3$, $y^2 x$. Then the 
differentials of the second kind satisfying \eqref{rCond} are
\begin{align*}
 &\begin{pmatrix} \rmd r_1(x,y) \\ \rmd r_2(x,y) \end{pmatrix}
 = \begin{pmatrix}
                   x^2 \\ 2 y x - \lambda_1 x^2
                  \end{pmatrix} \frac{\rmd x}{\partial_y f}.
\end{align*}
Then the expression for $\mathcal{R}_7$ acquires the form
\begin{align*}
\mathcal{R}_7 (x,y;u) &\equiv  2 y x - \lambda_1 x^2 - y \big(\wp_{1,2}(u) - \wp_{1,1,1}(u) \big) \\
&\qquad \quad - x \big(\wp_{2,2}(u) - \wp_{1,1,2}(u) \big) - \big(\wp_{2,5}(u) - \wp_{1,1,5}(u) \big).
\end{align*}
\end{rem}

\subsection{Example: $(3,7)$-Curve}
The family $\mathcal{V}_{(3,7)}$ of genus $6$ is defined by the equation
\begin{gather*}
\begin{split}
-y^3 + x^7 &+ y \big(\lambda_2 x^4 + \lambda_5 x^3 + \lambda_8 x^2 + \lambda_{11} x + \lambda_{14} \big) \\
&+ \lambda_6 x^5 + \lambda_9 x^4 + \lambda_{12} x^3 + \lambda_{15} x^2 + \lambda_{18} x + \lambda_{21} = 0
\end{split}
\end{gather*}
with the Weierstrass gap sequence $\{1,\, 2, \,4,\, 5,\, 8,\, 11\}$.
The first six monomials
$\M_{-1} = y x$, $\M_{-2} = x^3$,  $\M_{-4} = y$, $\M_{-5} = x^2$, $\M_{-8} = x$,  $\M_{-11} = 1$ from $\mathfrak{M}$
serve as numerators of differentials of the first kind.
The next two monomials $\M_{1}= x^4$,  $\M_{2}=y x^2$ produce differentials of the second kind
$\rmd r_1$  and $\rmd r_2$ of the form \eqref{Dif2s33m1}.

A solution of the Jacobi inversion problem for $D=\sum_{k=1}^6 (x_k,y_k)$ on $\mathcal{V}_{(3,7)}$ 
such that $u = \mathcal{A}(D)$ is given by the system
\begin{gather*}
\mathcal{R}_{12} (x,y;u) = 0,\qquad\qquad 
\mathcal{R}_{13} (x,y;u) = 0,
\end{gather*}
where
\begin{align*}
\mathcal{R}_{12} (x,y;u) &= x^4 - y x \wp_{1,1}(u) - x^3 \wp_{1,2}(u) - y \wp_{1,4}(u) \\
&\qquad \  - x^2 \wp_{1,5}(u) - x \wp_{1,8}(u) - \wp_{1,11}(u), \\
\mathcal{R}_{13} (x,y;u) &= 2 y x^2 - y x \big(\wp_{1,2}(u) - \wp_{1,1,1}(u) \big) - 
x^3 \big(\wp_{2,2}(u) - \wp_{1,1,2}(u) \big)  \\
&\qquad \quad \ - y \big(\wp_{2,4}(u) - \wp_{1,1,4}(u) \big) 
- x^2 \big(\wp_{2,5}(u) - \wp_{1,1,5}(u) \big) \\
&\qquad \quad \  - x \big(\wp_{2,8}(u) - \wp_{1,1,8}(u) \big) -
\big(\wp_{2,11}(u) - \wp_{1,1,11}(u) \big).
\end{align*}

\subsection{Proof of Theorem~\ref{T:C33m2} ($(3,3\mFr+2)$-Curves)}
Let $\xi$ be a local parameter in the vicinity of infinity on the curve \eqref{V33m2Eq}, 
then
\begin{equation}\label{xyParamC33m2}
 x(\xi) = \xi^{-3},\qquad y(\xi) = \xi^{-3\mFr-2}\Big(1 + \frac{\lambda_1}{3} \xi - 
 \frac{\lambda_1^3 }{81} \xi^3 + O(\xi^4)\Big).
\end{equation}
The basis of differentials of the first kind has the form
\begin{align}\label{Dif1s33m2}
 \rmd u(x,y) &= \frac{\rmd x}{\partial_y f}  \big(x^{2\mFr},\, y x^{\mFr-1},\, \dots,\, 
 x^{\mFr+1},\, y,\, x^{\mFr},\, \dots,\,x,\,1 \big)^t \\
  &= \frac{\rmd x}{\partial_y f}  \big(\M_{-1},\, \M_{-2},\, \dots,\, \M_{-(3\mFr-2)},\,  \M_{-(3\mFr-1)}, \notag\\
 &\phantom{mmmmmmm} \M_{-(3\mFr+1)},\, \dots,\, \M_{-(6\mFr-2)},\, \M_{-(6\mFr+1)} \big)^t, \notag
\end{align}
and the expansions near infinity are
\begin{align*}
&\rmd u_{3i-1} = y x^{\mFr-i} \frac{\rmd x}{\partial_y f} 
= \xi^{3i-2} \Big(1 + O(\xi^2)\Big) \rmd \xi,\quad i = 1,\, \dots,\, \mFr,\\
&\rmd u_{3i-2} = x^{2\mFr+1-i} \frac{\rmd x}{\partial_y f} 
= \xi^{3i-3}  \Big(1 - \frac{\lambda_1}{3} \xi + O(\xi^3)\Big) \rmd \xi,\quad i = 1,\, \dots,\, 2\mFr+1.
\end{align*}
By integration with respect to $\xi$ we find expansions for integrals of the first kind:
\begin{subequations}\label{Int1s33m2} 
\begin{align}
& u_{3i-1}(\xi) = \frac{\xi^{3i-1}}{3i-1} 
+ O(\xi^{3i+1}),\ i = 1,\, \dots,\, \mFr,\\
& u_{3i-2}(\xi) = \frac{\xi^{3i-2}}{3i-2} 
- \frac{\lambda_1}{3} \frac{\xi^{3i-1}}{3i-1} 
+ O(\xi^{3i+1}),\ i = 1,\, \dots,\, 2\mFr+1.
\end{align}
\end{subequations}

Let the two differentials of the second kind of the smallest  Sat\={o} weights be
\begin{align}\label{Dif2s33m2}
 &(\rmd r_1(x,y),\, \rmd r_2(x,y)) = \frac{\rmd x}{\partial_y f} 
 \big(yx^{\mFr},\, 2 x^{2\mFr+1}\big).
\end{align}
With the help of \eqref{xyParamC33m2} we find the corresponding expansions near infinity:
\begin{align*}
&\rmd r_1(\xi) = \xi^{-2} \Big(1 + O(\xi^2)\Big) \rmd \xi,\\
&\rmd r_2(\xi) = 2 \xi^{-3} \Big(1 - \frac{\lambda_1}{3} \xi + O(\xi^3)\Big) \rmd \xi.
\end{align*}
By a direct computation we find that $\rmd r_1$ satisfies the condition \eqref{rCond}, and $\rmd r_2$  does not. 
Thus, we replace the latter with
\begin{equation}\label{Dif2s33m2N}
 \rmd \widetilde{r}_2(x,y) =  (2 x^{2\mFr+1} + \lambda_1 y x^\mFr) \frac{\rmd x}{\partial_y f}, 
\end{equation}
which satisfies \eqref{rCond}.
The corresponding integrals of the second kind have the form
\begin{align*}
& r_1(\xi) = -\xi^{-1}  + O(\xi) + c_1,\\
& \widetilde{r}_2(\xi) = - \xi^{-2} - \frac{\lambda_1}{3} \xi^{-1} + O(\xi) + c_2.
\end{align*}

The integrals of the first kind \eqref{Int1s33m2} define $\mathcal{A}(\xi)$ from \eqref{rExpr}, namely
\begin{align*}
\mathcal{A}(\xi) &= \big(u_1(\xi),\, u_2(\xi),\, \dots,\, u_{3\mFr-2}(\xi),\, u_{3\mFr-1}(\xi),\\
&\qquad\qquad u_{3\mFr+1}(\xi),\, \dots,\,  u_{6\mFr-2}(\xi),\, u_{6\mFr+1}(\xi)\big)^t.
\end{align*}
Next, we compute residues on the right hand side of  \eqref{rExpr}.
Taking into account that $\mathcal{A}(0)=0$, $\mathcal{A}'(0)=(\delta_{i,1})$, 
$\mathcal{A}''(0)=(\delta_{i,2}-\frac{1}{3} \lambda_1 \delta_{i,1})$, 
where $\delta_{i,k}$ denotes the Kronecker delta and $i$ runs from $1$ to $g=3\mFr+1$, we find
\begin{equation*}
\frac{\rmd}{\rmd \xi} \log \sigma\big(u - \mathcal{A}(\xi)\big) = - \zeta_1(u)
- \big(\zeta_2(u) - \tfrac{1}{3} \lambda_1 \zeta_1(u) + \wp_{1,1}(u)\big) \xi + O(\xi^2).
\end{equation*}
Here $u = \mathcal{A}(D)$ is the Abel image of a non-special divisor 
$D = \sum_{k=1}^{3\mFr+1} (x_k,y_k)$.
Thus, \eqref{rExpr} acquires the form
\begin{gather}\label{ZetaC33m2}
 \sum_{k=1}^{3\mFr+1} \begin{pmatrix}
r_1(x_k,\,y_k) \\ \widetilde{r}_2(x_k,\,y_k) 
\end{pmatrix} = - \begin{pmatrix} \zeta_1(u) \\
\zeta_2(u)  + \wp_{1,1}(u) \end{pmatrix}
\equiv \begin{pmatrix} R_1 (u) \\  R_2 (u)\end{pmatrix}.
\end{gather}
Differentiating the above relations with respect to $x_1$, we find the relations
\begin{subequations}\label{RrelsC33m2}
\begin{align}
&y_1 x_1^{\mFr} 
= \bigg( y_1 \sum_{i=1}^{\mFr} \wp_{1,3i-1}(u) x_1^{\mFr-i} 
+ \sum_{i=1}^{2\mFr+1} \wp_{1,3i-2}(u) x_1^{2\mFr+1-i} \bigg),\\
&2 x_1^{2\mFr+1} + \lambda_1 y_1 x_1^{\mFr} = 
\bigg( y_1 \sum_{i=1}^{\mFr} \big(\wp_{2,3i-1}(u) - \wp_{1,1,3i-1}(u) \big) x_1^{\mFr-i} \\
&\phantom{mmmmmmmmmm} + \sum_{i=1}^{2\mFr+1} \big( \wp_{2,3i-2}(u) 
- \wp_{1,1,3i-2}(u) \big) x_1^{2\mFr+1-i} \bigg), \notag
\end{align} 
\end{subequations}
which coincide with
\eqref{REqsC33m2}. $\qede$

\begin{rem}\label{R:Dif2}
The differentials of the second kind $\rmd r_1$ defined by \eqref{Dif2s33m2}  and $\rmd \widetilde{r}_2$
defined by \eqref{Dif2s33m2N} on $\mathcal{V}_{(3,3\mFr+2)}$
are associated with the differentials of the first kind \eqref{Dif1s33m2}, according to \cite[Art.\;138]{bakerAF}.
Recall that differentials of the second kind $\rmd r$ and of the first kind $\rmd u$ form an associated system 
if the algebraic part of the fundamental bi-differential, see for example \cite[sect.\;10.1--10.2]{belMDSF}, 
in the form \cite[Eq.\;(1.9)]{belHKF} is symmetric.
The integrals of the second kind $r_{\wFr_i}$ obtained from the differentials associated with differentials of the first kind
are expressed through the corresponding zeta functions $\zeta_{\wFr_i}$ defined by \eqref{zetaDef} 
with Abelian functions added, see \eqref{ZetaC33m2} and \eqref{ZetaC33m1}. 
As a result, relations of the type \eqref{RrelsC33m2}
have the simplest form.
\end{rem}

\subsection{Example: $(3,5)$-Curve}
The family $\mathcal{V}_{(3,5)}$ of genus $4$ is defined by the equation
\begin{equation*}
-y^3 + x^5 + y (\lambda_1 x^3 + \lambda_4 x^2 + \lambda_7 x + \lambda_{10})
+ \lambda_6 x^3 + \lambda_9 x^2 + \lambda_{12} x + \lambda_{15} = 0,
\end{equation*}
and the Weierstrass sequence is $\{1,\, 2, \,4,\, 7\}$. The first four monomials 
in the list $\mathfrak{M}$ sorted in the ascending order of the Sat\={o} weight are
$\M_{-1} = x^2$, $\M_{-2} = y$, $\M_{-4} = x$,  $\M_{-7} = 1$; they serve as numerators of differentials of the first kind.
The next two monomials are $\M_{1}= y x$,  $\M_{2} = x^3$; they produce differentials of the second kind:
$\rmd r_1 = y x \rmd x/(\partial_y f)$  and $\rmd r_2 = (2 x^3 + \lambda_1 y x) \rmd x/(\partial_y f)$.

A solution of the Jacobi inversion problem for $D=\sum_{k=1}^4 (x_k,y_k)$ on $\mathcal{V}_{(3,5)}$ 
such that $u = \mathcal{A} (D)$
 is given by the system
\begin{align*}
0=\mathcal{R}_8 (x,y;u) &\equiv y x - \wp_{1,1}(u) x^2 - \wp_{1,2}(u) y - \wp_{1,4}(u) x - \wp_{1,7}(u), \\
0=\mathcal{R}_9 (x,y;u) &\equiv 2 x^3 + \lambda_1 y x - \big(\wp_{1,2}(u) - \wp_{1,1,1}(u) \big) x^2
- \big(\wp_{2,2}(u) - \wp_{1,1,2}(u) \big) y \\
&\qquad \qquad \qquad\ - \big(\wp_{2,4}(u) - \wp_{1,1,4}(u) \big) x -
\big(\wp_{2,7}(u) - \wp_{1,1,7}(u) \big).
\end{align*}

\section{Jacobi inversion problem on tetragonal curves}
There exist two types of tetragonal $(n,s)$-curves:  
$(4,4\mFr+1)$, $(4,4\mFr+3)$, where $\mFr$ is a natural number.

\begin{theo}[$(4,4\mFr+1)$-Curves]\label{T:C44m1}
Let a non-degenerate $(4,4\mFr+1)$-curve of genus $g=6\mFr$ be defined by
\begin{equation}\label{V44m1Eq}
-y^4 + x^{4\mFr+1} + y^2 \sum_{i=0}^{2\mFr} \lambda_{4i+2} x^{2\mFr-i} 
+ y \sum_{i=0}^{3\mFr} \lambda_{4i+3} x^{3\mFr-i} 
+ \sum_{i=1}^{4\mFr } \lambda_{4i+4} x^{4\mFr-i} = 0.
\end{equation}
Let $u = \mathcal{A}(D)$ be the Abel image of  a non-special degree $g$ divisor  $D$
on the curve. Then $D$ is uniquely defined by the system of equations 
\begin{subequations}\label{REqsC44m1}
\begin{gather}
\mathcal{R}_{12\mFr}(x,y;u)=0,\qquad \mathcal{R}_{12\mFr+1}(x,y;u)=0,\qquad
\mathcal{R}_{12\mFr+2}(x,y;u)=0
\end{gather}
with three entire rational functions of the weights $2g=12\mFr$, $2g+1=12\mFr+1$,
$2g+2=12\mFr+2$:
\begin{align}
\mathcal{R}_{12\mFr}(x,y;u) 
&= x^{3\mFr} - \sum_{i=1}^{6\mFr} \wp_{1,\wFr_{i}}(u) \M_{-\wFr_{i}}, \\
\mathcal{R}_{12\mFr+1}(x,y;u) 
&= 2 y x^{2\mFr}  - \sum_{i=1}^{6\mFr} \big( \wp_{2,\wFr_{i}}(u) - \wp_{1,1,\wFr_{i}}(u) \big)  \M_{-\wFr_{i}}, \\
\mathcal{R}_{12\mFr+2}(x,y;u) 
 &= 3 y^2 x^{\mFr} - \lambda_2 x^{3\mFr} \\
&\hspace{-10mm} - \sum_{i=1}^{6\mFr} \big(\wp_{3,\wFr_{i}}(u) - \tfrac{3}{2} \wp_{1,2,\wFr_{i}}(u)
 + \tfrac{1}{2} \wp_{1,1,1,\wFr_{i}}(u) \big)  \M_{-\wFr_{i}}, \notag 
 \end{align}
\end{subequations}
where
\begin{align*}
&\M_{-(4i-3)} = y^2 x^{\mFr-i},\quad  i=1,\, \dots,\, \mFr, \\
&\M_{-(4i-2)} = y x^{2\mFr-i},\quad  i=1,\, \dots,\, 2\mFr, \\
&\M_{-(4i-1)} = x^{3\mFr-i},\quad  i=1,\, \dots,\, 3\mFr.
\end{align*}
\end{theo}

\begin{theo}[$(4,4\mFr+3)$-Curves]\label{T:C44m3}
Let a non-degenerate $(4,4\mFr+3)$-curve of genus $g=6\mFr+3$ be defined by
\begin{align}\label{V44m3Eq}
-y^4 + x^{4\mFr+3} &+ y^2 \sum_{i=0}^{2\mFr+1} \lambda_{4i+2} x^{2\mFr+1-i} \\
&+ y \sum_{i=0}^{3\mFr+2} \lambda_{4i+1} x^{3\mFr+2-i} 
+ \sum_{i=1}^{4\mFr +2} \lambda_{4i+4} x^{4\mFr+2-i} = 0. \notag
\end{align}
Let $u = \mathcal{A}(D)$ be the Abel image of  a non-special degree $g$ divisor  $D$
on the curve. Then $D$ is uniquely defined by the system of equations 
\begin{subequations}\label{REqsC44m3}
\begin{gather}
\mathcal{R}_{12\mFr+6}(x,y;u)=0,\qquad \mathcal{R}_{12\mFr+7}(x,y;u)=0,\qquad
\mathcal{R}_{12\mFr+8}(x,y;u)=0
\end{gather}
with three entire rational functions of the weights $2g=12\mFr+6$, $2g+1=12\mFr+7$,
$2g+2=12\mFr+8$:
\begin{align}
&\mathcal{R}_{12\mFr+6}(x,y;u) 
= y^2 x^{\mFr} - \sum_{i=1}^{6\mFr+3} \wp_{1,\wFr_{i}}(u) \M_{-\wFr_{i}}, \\
&\mathcal{R}_{12\mFr+7}(x,y;u) 
= 2 y x^{2\mFr+1} \,{+}\, \lambda_1  y^2 x^{\mFr}  
\,{-}\, \sum_{i=1}^{6\mFr+3} \big( \wp_{2,\wFr_{i}}(u) - \wp_{1,1,\wFr_{i}}(u) \big)  \M_{-\wFr_{i}}, \\
&\mathcal{R}_{12\mFr+8}(x,y;u) 
 = 3x^{3\mFr+2} + 2\lambda_1 y x^{2\mFr+1}  + \lambda_2 y^2 x^{\mFr}  \\
&\hspace{3mm} - \sum_{i=1}^{6\mFr+3} \big(\wp_{3,\wFr_{i}}(u) - \tfrac{3}{2} \wp_{1,2,\wFr_{i}}(u)
 + \tfrac{1}{2} \lambda_1 \wp_{1,1,\wFr_{i}}(u) + \tfrac{1}{2} \wp_{1,1,1,\wFr_{i}}(u) \big)  \M_{-\wFr_{i}}, \notag 
 \end{align}
\end{subequations}
where
\begin{align*}
&\M_{-(4i-1)} = y^2 x^{\mFr-i},\quad  i=1,\, \dots,\, \mFr, \\
&\M_{-(4i-2)} = y x^{2\mFr+1-i},\quad  i=1,\, \dots,\, 2\mFr+1, \\
&\M_{-(4i-3)} = x^{3\mFr+2-i},\quad  i=1,\, \dots,\, 3\mFr+2.
\end{align*}
\end{theo}

\subsection{Proof of Theorem~\ref{T:C44m1} ($(4,4\mFr+1)$-Curves)}
Let $\xi$ be a local coordinate in the vicinity of infinity on the curve \eqref{V44m1Eq}, then
\begin{equation}\label{xyParamC44m1}
 x(\xi) = \xi^{-4},\qquad y(\xi) = \xi^{-4\mFr-1}\bigg(1 + \frac{\lambda_2}{4} \xi^2 
 + \frac{\lambda_3}{4} \xi^3 + \frac{\lambda_2^2}{32}  \xi^4 + O(\xi^6)\bigg).
\end{equation}
The basis of differentials of the first kind has the form
\begin{align*}
 \rmd u(x,y) = \frac{\rmd x}{\partial_y f}  \big(y^2 x^{\mFr-1},\, y x^{2\mFr-1},\, x^{3\mFr-1},\, 
 &\dots,\, y^2,\, y x^{\mFr},\, x^{2\mFr},\, \\
 &\dots,\, y,\, x^{\mFr},\, \dots,\,x,\,1 \big)^t. \notag
\end{align*}
and the expansions near infinity are
\begin{align*}
&\rmd u_{4i-3} = y^2 x^{\mFr-i} \frac{\rmd x}{\partial_y f} 
=\xi^{4i-4}  \Big(1 + \frac{\lambda_2}{4} \xi^2 + O(\xi^4)\Big) \rmd \xi,\quad i = 1,\, \dots,\, \mFr,\\
&\rmd u_{4i-2} = y x^{2\mFr-i} \frac{\rmd x}{\partial_y f} 
=\xi^{4i-3}  \Big(1 + O(\xi^3)\Big) \rmd \xi,\quad i = 1,\, \dots,\, 2\mFr,\\
&\rmd u_{4i-1} = x^{3\mFr-i} \frac{\rmd x}{\partial_y f} 
=\xi^{4i-2}  \Big(1 + O(\xi^2)\Big) \rmd \xi,\quad i = 1,\, \dots,\, 3\mFr.
\end{align*}
By integration with respect to $\xi$ we find expansions for integrals of the first kind:
\begin{subequations}\label{Int1s44m1} 
\begin{align}
& u_{4i-3}(\xi) = \frac{\xi^{4i-3}}{4i-3} 
+ \frac{\lambda_2}{4} \frac{\xi^{4i-1}}{4i-1} 
+ O(\xi^{4i+1}),\ i = 1,\, \dots,\, \mFr, \\
& u_{4i-2}(\xi) = \frac{\xi^{4i-2}}{4i-2} 
+ O(\xi^{4i+1}),\ i = 1,\, \dots,\, 2\mFr, \\
& u_{4i-1}(\xi) = \frac{\xi^{4i-1}}{4i-1} 
+ O(\xi^{4i+1}),\ i = 1,\, \dots,\, 3\mFr.
\end{align}
\end{subequations}

Let the three differentials of the second kind of weights $1$, $2$, $3$ be
\begin{align}\label{Dif2C44m1}
 &\big(\rmd r_1(x,y),\, \rmd r_2(x,y),\, \rmd r_3(x,y) \big) = \frac{\rmd x}{\partial_y f} 
 \big(x^{3\mFr},\, 2 yx^{2\mFr},\, 3 y^2 x^{\mFr}\big). 
\end{align}
and the corresponding expansions near infinity
\begin{align*}
&\rmd r_1(\xi) = \xi^{-2} \Big(1 + O(\xi^2)\Big) \rmd \xi,\\
&\rmd r_2(\xi) = 2 \xi^{-3} \Big(1 + O(\xi^3)\Big) \rmd \xi,\\
&\rmd r_3(\xi) = 3 \xi^{-4}\Big(1 + \frac{\lambda_2}{4} \xi^2 + O(\xi^4)\Big) \rmd \xi.
\end{align*}
In order to satisfy \eqref{rCond}, we replace $\rmd r_3$ with
\begin{equation}\label{Dif3s44m1N}
 \rmd \widetilde{r}_3(x,y) =  (3 y^2 x^{\mFr} - \lambda_2 x^{3\mFr}) \frac{\rmd x}{\partial_y f}.
\end{equation}
The differentials of the second kind $\rmd r_1$, $\rmd r_2$, $\rmd \widetilde{r}_3$  form an
associated system with differentials of the first kind. 
The corresponding integrals of the second kind have the form
\begin{align*}
& r_1(\xi) = -\xi^{-1}  + O(\xi) + c_1,\\
& r_2(\xi) = - \xi^{-2}  + O(\xi) + c_2,\\
& \widetilde{r}_3(\xi) = - \xi^{-3} + \frac{\lambda_2}{4} \xi^{-1} + O(\xi) + c_3.
\end{align*}

The integrals of the first kind \eqref{Int1s44m1} define $\mathcal{A}(\xi)$ from \eqref{rExpr}, namely
\begin{align*}
\mathcal{A}(\xi) = \big(u_1(\xi),\, u_2(\xi),\, &u_3(\xi),\, \dots,\,  u_{4\mFr-3}(\xi),\, u_{4\mFr-2}(\xi),\,
u_{4\mFr-1}(\xi),\, \dots,\,  \\
& u_{8\mFr-2}(\xi),\, u_{8\mFr-1}(\xi),\, \dots,\, u_{12\mFr-5}(\xi),\, u_{12\mFr-1}(\xi)\big)^t.
\end{align*}
Next, we compute residues on the right hand side of \eqref{rExpr}.
Taking into account that $\mathcal{A}(0)=0$, $\mathcal{A}'(0)=(\delta_{i,1})$, 
$\mathcal{A}''(0)=(\delta_{i,2})$, $\mathcal{A}^{(3)}(0)=(2\delta_{i,3}+\frac{1}{2} \lambda_2 \delta_{i,1})$, 
where $\delta_{i,k}$ denotes the Kronecker delta and $i$ runs from $1$ to $g = 6\mFr$, we find
\begin{multline*}
\frac{\rmd}{\rmd \xi} \log \sigma\big(u - \mathcal{A}(\xi)\big) = - \zeta_1(u)
- \big(\zeta_2(u) + \wp_{1,1}(u)\big) \xi \\
- \big( \zeta_3(u) + \tfrac{1}{4}\lambda_2  \zeta_1(u) + \tfrac{3}{2} \wp_{1,2}(u) 
- \tfrac{1}{2} \wp_{1,1,1}(u) \big) \xi^2 + O(\xi^3).
\end{multline*}
Here $u = \mathcal{A}(D)$ is the Abel's image of a non-special divisor 
$D = \sum_{k=1}^{6\mFr} (x_k,y_k)$.
Thus, \eqref{rExpr} acquires the form
\begin{align*}
& \sum_{k=1}^{6\mFr} \begin{pmatrix}
r_1(x_k,\,y_k) \\ r_2(x_k,\,y_k) \\ \widetilde{r}_3(x_k,\,y_k) 
\end{pmatrix} = - \begin{pmatrix} \zeta_1(u) \\
\zeta_2(u) + \wp_{1,1}(u) \\
\zeta_3(u) + \tfrac{3}{2} \wp_{1,2}(u) - \tfrac{1}{2} \wp_{1,1,1}(u)
\end{pmatrix} \equiv \begin{pmatrix} R_1 (u) \\  R_2 (u) \\ R_3 (u) \end{pmatrix}.
\end{align*}
Differentiating the above relations with respect to $x_1$, we find the relations
\begin{align*}
&x_1^{3\mFr} 
=  y_1^2 \sum_{i=1}^{\mFr} \wp_{1,4i-3} x_1^{\mFr-i}+
y_1 \sum_{i=1}^{2\mFr} \wp_{1,4i-2}x_1^{2\mFr-i} 
+ \sum_{i=1}^{3\mFr} \wp_{1,4i-1} x_1^{3\mFr-i}, \\
&2y_1 x_1^{2\mFr}  = 
 y_1^2 \sum_{i=1}^{\mFr} \big(\wp_{2,4i-3} - \wp_{1,1,4i-3} \big) x_1^{\mFr-i} 
 + y_1 \sum_{i=1}^{2\mFr} \big(\wp_{2,4i-2} - \wp_{1,1,4i-2} \big) x_1^{2\mFr-i} \\
&\phantom{mmmmm} + \sum_{i=1}^{3\mFr} \big( \wp_{2,4i-1} - \wp_{1,1,4i-1} \big) x_1^{3\mFr-i},\\
&3 y_1^2 x_1^{\mFr} - \lambda_2 x_1^{3\mFr} = 
 y_1^2 \sum_{i=1}^{\mFr} \big(\wp_{3,4i-3}- \tfrac{3}{2} \wp_{1,2,4i-3}
+ \tfrac{1}{2} \wp_{1,1,1,4i-3}\big) x_1^{\mFr-i} \\
&\phantom{mmmmmmmmm} + y_1 \sum_{i=1}^{2\mFr} \big(\wp_{3,4i-2}- \tfrac{3}{2} \wp_{1,2,4i-2}
+ \tfrac{1}{2} \wp_{1,1,1,4i-2}\big) x_1^{2\mFr-i} \\
&\phantom{mmmmmmmmm} + \sum_{i=1}^{3\mFr} \big(\wp_{3,4i-1}- \tfrac{3}{2} \wp_{1,2,4i-1}
+ \tfrac{1}{2} \wp_{1,1,1,4i-1}\big) x_1^{3\mFr-i},
\end{align*} 
which coincide with \eqref{REqsC44m1}. 
The argument of abelian functions is omitted, that is $\wp_{i,j}$,
$\wp_{i,j,k}$, $\wp_{i,j,k,l}$ stand for $\wp_{i,j}(u)$, $\wp_{i,j,k}(u)$, $\wp_{i,j,k,l}(u)$.
$\qede$

\subsection{Proof of Theorem~\ref{T:C44m3} ($(4,4\mFr+3)$-Curves)}
Let $\xi$ be a local coordinate in the vicinity of infinity on the curve \eqref{V44m3Eq}, then
\begin{equation}\label{xyParamC44m3}
 x(\xi) = \xi^{-4},\qquad y(\xi) = \xi^{-4\mFr-3}\bigg(1 + \frac{\lambda_1}{4} \xi 
 + \Big( \frac{\lambda_2 }{4} - \frac{\lambda_1^2}{32} \Big) \xi^2 + O(\xi^4)\bigg).
\end{equation}
The basis of differentials of the first kind has the form
\begin{align*}
 \rmd u(x,y) = \frac{\rmd x}{\partial_y f}  \big(x^{3\mFr+1},\,  y x^{2\mFr},\,  y^2 x^{\mFr-1},\, 
 &\dots,\, y^2,\, x^{2\mFr+1},\, y x^{\mFr},\, \\
 &\dots,\, y,\, x^{\mFr},\, \dots,\,x,\,1 \big)^t. \notag
\end{align*}
and the expansions near infinity are
\begin{align*}
&\rmd u_{4i-1} = y^2 x^{\mFr-i} \frac{\rmd x}{\partial_y f} 
=\xi^{4i-2}  \Big(1 + O(\xi^2)\Big) \rmd \xi,
\quad  i = 1,\, \dots,\, \mFr,\\
&\rmd u_{4i-2} = y x^{2\mFr+1-i} \frac{\rmd x}{\partial_y f} 
=\xi^{4i-3}  \Big(1 - \frac{\lambda_1}{4} \xi + O(\xi^3)\Big) \rmd \xi,\quad i = 1,\, \dots,\, 2\mFr+1, \\
&\rmd u_{4i-3} = x^{3\mFr+2-i} \frac{\rmd x}{\partial_y f} 
=\xi^{4i-4}  \Big(1 - \frac{\lambda_1}{2} \xi 
- \Big(\frac{\lambda_2}{4} - \frac{5\lambda_1^2}{32} \Big) \xi^2 + O(\xi^4)\Big) \rmd \xi,\\
&\qquad i = 1,\, \dots,\, 3\mFr+2.
\end{align*}
By integration with respect to $\xi$ we find expansions for integrals of the first kind:
\begin{subequations}\label{Int1s44m3} 
\begin{align}
& u_{4i-1}(\xi) = \frac{\xi^{4i-1}}{4i-1} + O(\xi^{4i+1}),\ \ i = 1,\, \dots,\, \mFr,\\
& u_{4i-2}(\xi) = \frac{\xi^{4i-2}}{4i-2} 
- \frac{\lambda_1}{4} \frac{\xi^{4i-1}}{4i-1} 
+ O(\xi^{4i+1}),\ i = 1,\, \dots,\, 2\mFr+1, \\
& u_{4i-3}(\xi) = \frac{\xi^{4i-3}}{4i-3} 
- \frac{\lambda_1}{2} \frac{\xi^{4i-2}}{4i-2} 
-  \frac{8 \lambda_2 - 5 \lambda_1^2}{32} \frac{\xi^{4i-1}}{4i-1} 
+ O(\xi^{4i+1}), \\
&\quad i = 1,\, \dots,\, 3\mFr+2. \notag
\end{align}
\end{subequations}

Let the three differentials of the second kind of weights $1$, $2$, $3$ be
\begin{align}
 &\big(\rmd r_1(x,y),\, \rmd r_2(x,y),\, \rmd r_3(x,y) \big) = \frac{\rmd x}{\partial_y f} 
 \big(y^2 x^{\mFr},\, 2 yx^{2\mFr+1},\,  3x^{3\mFr+2}\big). 
\end{align}
and the corresponding expansions near infinity
\begin{align*}
&\rmd r_1(\xi) = \xi^{-2} \Big(1  + O(\xi^2)\Big) \rmd \xi,\\
&\rmd r_2(\xi) = 2 \xi^{-3} \Big(1 - \frac{\lambda_1}{4} \xi + O(\xi^3) \Big) \rmd \xi,\\
&\rmd r_3(\xi) = 3 \xi^{-4}\Big(1 - \frac{\lambda_1}{2} \xi 
- \Big(\frac{\lambda_2}{4} - \frac{5\lambda_1^2}{32} \Big) \xi^2 + O(\xi^4)\Big) \rmd \xi.
\end{align*}
With the help of \eqref{rCond}, we find the differentials of the second kind 
which form an associated system with differentials of the first kind:
\begin{align*}
\begin{pmatrix}
\rmd r_1(x,y) \\
\rmd \widetilde{r}_2(x,y) \\
\rmd \widetilde{r}_3(x,y) 
\end{pmatrix} =  \frac{\rmd x}{\partial_y f}
\begin{pmatrix}
y^2 x^{\mFr} \\
2 yx^{2\mFr+1} + \lambda_1 y^2 x^{\mFr} \\
3x^{3\mFr+2} + 2 \lambda_1 yx^{2\mFr+1} + \lambda_2 y^2 x^{\mFr}
\end{pmatrix}.
\end{align*}
By integration with respect to $\xi$ we obtain
\begin{align*}
& r_1(\xi) = -\xi^{-1}  + O(\xi) + c_1,\\
& \widetilde{r}_2(\xi) = - \xi^{-2}  - \frac{\lambda_1}{2} \xi^{-1} + O(\xi) + c_2,\\
& \widetilde{r}_3(\xi) = - \xi^{-3} - \frac{\lambda_1}{4}  \xi^{-2} 
- \frac{8 \lambda_2 - \lambda_1^2}{32} \xi^{-1} + O(\xi) + c_3.
\end{align*}

The integrals of the first kind \eqref{Int1s44m3} define $\mathcal{A}(\xi)$ from \eqref{rExpr}, namely
\begin{align*}
\mathcal{A}(\xi) = \big(u_1(\xi),\, u_2(\xi),\, &u_3(\xi),\, \dots,\,  u_{4\mFr-1}(\xi),\, u_{4\mFr+1}(\xi),\,
u_{4\mFr+2}(\xi),\, \dots,\,  \\
& u_{8\mFr+2}(\xi),\, u_{8\mFr+5}(\xi),\, \dots,\, u_{12\mFr+1}(\xi),\, u_{12\mFr+5}(\xi)\big)^t.
\end{align*}
Next, we compute residues on the right hand side of \eqref{rExpr}.
Taking into account that $\mathcal{A}(0)=0$, $\mathcal{A}'(0)=(\delta_{i,1})$, 
$\mathcal{A}''(0)=(\delta_{i,2} - \tfrac{1}{2} \lambda_1 \delta_{i,1})$, 
$\mathcal{A}^{(3)}(0)=(2\delta_{i,3}- \tfrac{1}{2} \lambda_1 \delta_{i,2} 
- (\frac{1}{2} \lambda_2 - \tfrac{5}{16} \lambda_1^2) \delta_{i,1})$, 
where $\delta_{i,k}$ denotes the Kronecker delta and $i$ runs from $1$ to $g = 6\mFr+3$, we find
\begin{align*}
\frac{\rmd}{\rmd \xi} \log \sigma\big(u - \mathcal{A}(\xi)\big) &= - \zeta_1(u)
- \big(\zeta_2(u)  - \tfrac{1}{2} \lambda_1\zeta_1(u) + \wp_{1,1}(u)\big) \xi \\
&- \big( \zeta_3(u) - \tfrac{1}{4}\lambda_1  \zeta_2(u) - \big(\tfrac{1}{4}\lambda_2 - \tfrac{5}{32} \lambda_1^2 \big)  \zeta_1(u) \\
&\qquad + \tfrac{3}{2} \wp_{1,2}(u) - \tfrac{3}{4} \lambda_1 \wp_{1,1}(u)  - \tfrac{1}{2} \wp_{1,1,1}(u) \big) \xi^2 + O(\xi^3).
\end{align*}
Here $u = \mathcal{A}(D)$ is the Abel's image of a non-special divisor 
$D = \sum_{k=1}^{6\mFr+3} (x_k,y_k)$.
Thus, \eqref{rExpr} acquires the form
\begin{align*}
& \sum_{k=1}^{6\mFr+3} \begin{pmatrix}
r_1(x_k,\,y_k) \\ \widetilde{r}_2(x_k,\,y_k) \\ \widetilde{r}_3(x_k,\,y_k) 
\end{pmatrix} = - \begin{pmatrix} \zeta_1(u) \\
\zeta_2(u) + \wp_{1,1}(u) \\
\zeta_3(u) + \tfrac{3}{2} \wp_{1,2}(u) - \tfrac{1}{2} \lambda_1 \wp_{1,1}(u) - \tfrac{1}{2} \wp_{1,1,1}(u)
\end{pmatrix}  \equiv \begin{pmatrix} R_1 (u) \\  R_2 (u) \\ R_3 (u) \end{pmatrix}.
\end{align*}
Differentiating the above relations with respect to $x_1$, we find the relations
\eqref{REqsC44m3}. $\qede$

\section{Jacobi inversion problem on pentagonal curves}
There exist four types of pentagonal $(n,s)$-curves:  
$(5,5\mFr+1)$, $(5,5\mFr+2)$, $(5,5\mFr+3)$, and $(5,5\mFr+4)$, where $\mFr$ is a natural number.

\begin{theo}[$(5,5\mFr+1)$-Curves]\label{T:C55m1}
Let a non-degenerate $(5,5\mFr+1)$-curve of genus $g=10\mFr$ be defined by
\begin{align}\label{V55m1Eq}
-y^5 + x^{5\mFr+1} + & y^3 \sum_{i=0}^{2\mFr} \lambda_{5i+2} x^{2\mFr-i} 
+ y^2 \sum_{i=0}^{3\mFr} \lambda_{5i+3} x^{3\mFr-i} \\
&+ y \sum_{i=0}^{4\mFr} \lambda_{5i+4} x^{4\mFr-i} 
+ \sum_{i=1}^{5\mFr } \lambda_{5i+5} x^{5\mFr-i} = 0. \notag
\end{align}
Let $u = \mathcal{A}(D)$ be the Abel image of  a non-special degree $g$ divisor  $D$
on the curve. Then $D$ is uniquely defined by the system of equations 
\begin{subequations}\label{REqsC55m1}
\begin{gather}
\mathcal{R}_{20\mFr + l}(x,y;u)=0,\qquad l=0,1,2,3
\end{gather}
with four entire rational functions of the weights $2g + l =20\mFr + l$, $l=0,1,2,3$:
\begin{align}
\mathcal{R}_{20\mFr}(x,y;u) 
&= x^{4\mFr} - \sum_{i=1}^{10\mFr} \wp_{1,\wFr_{i}}(u) \M_{-\wFr_{i}}, \\
\mathcal{R}_{20\mFr+1}(x,y;u) 
&= 2 y x^{3\mFr}  - \sum_{i=1}^{10\mFr} \big( \wp_{2,\wFr_{i}}(u) - \wp_{1,1,\wFr_{i}}(u) \big)  \M_{-\wFr_{i}}, \\
\mathcal{R}_{20\mFr+2}(x,y;u) 
 &= 3 y^2 x^{2\mFr} - \lambda_2 x^{4\mFr} \\
&\hspace{-10mm} - \sum_{i=1}^{10\mFr} \big(\wp_{3,\wFr_{i}}(u) - \tfrac{3}{2} \wp_{1,2,\wFr_{i}}(u)
 + \tfrac{1}{2} \wp_{1,1,1,\wFr_{i}}(u) \big)  \M_{-\wFr_{i}}, \notag \\
\mathcal{R}_{20\mFr+3}(x,y;u) 
 &= 4 y^3 x^{\mFr} - 2 \lambda_2 y x^{3\mFr} - \lambda_3 x^{4\mFr} \\
&\hspace{-15mm} - \sum_{i=1}^{10\mFr} \big(\wp_{4,\wFr_{i}}(u) - \tfrac{1}{2} \wp_{2,2,\wFr_{i}}(u)
- \tfrac{4}{3} \wp_{1,3,\wFr_{i}}(u) - \tfrac{1}{3} \lambda_2 \wp_{1,1,\wFr_{i}}(u) \notag \\
& + \wp_{1,1,2,\wFr_{i}}(u) - \tfrac{1}{6}  \wp_{1,1,1,1,\wFr_{i}}(u) \big)  \M_{-\wFr_{i}}, \notag  
 \end{align}
\end{subequations}
where
\begin{align*}
&\M_{-(5i-4)} = y^3 x^{\mFr-i},\quad  i=1,\, \dots,\, \mFr, \\
&\M_{-(5i-3)} = y^2 x^{2\mFr-i},\quad  i=1,\, \dots,\, 2\mFr, \\
&\M_{-(5i-2)} = y x^{3\mFr-i},\quad  i=1,\, \dots,\, 3\mFr, \\
&\M_{-(5i-1)} = x^{4\mFr-i},\quad  i=1,\, \dots,\, 4\mFr.
\end{align*}
\end{theo}

\begin{theo}[$(5,5\mFr+2)$-Curves]\label{T:C55m2}
Let a non-degenerate $(5,5\mFr+2)$-curve of genus $g=10\mFr+2$ be defined by
\begin{align}\label{V55m2Eq}
-y^5 + x^{5\mFr+2} + & y^3 \sum_{i=0}^{2\mFr} \lambda_{5i+4} x^{2\mFr-i} 
+ y^2 \sum_{i=0}^{3\mFr+1} \lambda_{5i+1} x^{3\mFr+1-i} \\
&+ y \sum_{i=0}^{4\mFr+1} \lambda_{5i+3} x^{4\mFr+1-i} 
+ \sum_{i=1}^{5\mFr +1} \lambda_{5i+5} x^{5\mFr+1-i} = 0. \notag
\end{align}
Let $u = \mathcal{A}(D)$ be the Abel image of  a non-special degree $g$ divisor  $D$
on the curve. Then $D$ is uniquely defined by the system of equations 
\begin{subequations}\label{REqsC55m2}
\begin{gather}
\mathcal{R}_{20\mFr+4+ l}(x,y;u)=0,\qquad l=0,1,2,3
\end{gather}
with four entire rational functions of the weights $2g + l =20\mFr +4 + l$, $l=0,1,2,3$:
\begin{align}
\mathcal{R}_{20\mFr+4}(x,y;u) 
&= y^2 x^{2\mFr} - \sum_{i=1}^{10\mFr+2} \wp_{1,\wFr_{i}}(u) \M_{-\wFr_{i}}, \\
\mathcal{R}_{20\mFr+5}(x,y;u) 
&= 2 x^{4\mFr+1}  + \lambda_1 y^2 x^{2\mFr} \\
&\quad - \sum_{i=1}^{10\mFr+2} 
\big( \wp_{2,\wFr_{i}}(u) - \wp_{1,1,\wFr_{i}}(u) \big)  \M_{-\wFr_{i}}, \notag \\
\mathcal{R}_{20\mFr+6}(x,y;u) 
 &= 3 y^3 x^{\mFr} - \lambda_1 x^{4\mFr+1}\\
&\hspace{-20mm}  - \sum_{i=1}^{10\mFr+2} \big(\wp_{3,\wFr_{i}}(u) 
- \tfrac{3}{2} \wp_{1,2,\wFr_{i}}(u) + \tfrac{1}{2} \lambda_1 \wp_{1,1,\wFr_{i}}(u)
 + \tfrac{1}{2} \wp_{1,1,1,\wFr_{i}}(u) \big)  \M_{-\wFr_{i}}, \notag \\
\mathcal{R}_{20\mFr+7}(x,y;u) 
 &= 4 y x^{3\mFr+1} + 2 \lambda_1 y^3 x^{\mFr} + 2 \lambda_3 y^2 x^{2\mFr} \\
&\hspace{-20mm} - \sum_{i=1}^{10\mFr+2} \big(\wp_{4,\wFr_{i}}(u) - \tfrac{1}{2} \wp_{2,2,\wFr_{i}}(u)
- \tfrac{4}{3} \wp_{1,3,\wFr_{i}}(u) - \tfrac{2}{3} \lambda_1 \wp_{1,2,\wFr_{i}}(u) \notag \\
&  + \tfrac{1}{6} \lambda_1^2 \wp_{1,1,\wFr_{i}}(u)  + \wp_{1,1,2,\wFr_{i}}(u) 
- \tfrac{1}{6}  \wp_{1,1,1,1,\wFr_{i}}(u) \big)  \M_{-\wFr_{i}}, \notag  
 \end{align}
\end{subequations}
where
\begin{align*}
&\M_{-(5i-3)} = y^3 x^{\mFr-i},\quad  i=1,\, \dots,\, \mFr, \\
&\M_{-(5i-1)} = y^2 x^{2\mFr-i},\quad  i=1,\, \dots,\, 2\mFr, \\
&\M_{-(5i-4)} = y x^{3\mFr+1-i},\quad  i=1,\, \dots,\, 3\mFr+1, \\
&\M_{-(5i-2)} = x^{4\mFr+1-i},\quad  i=1,\, \dots,\, 4\mFr+1.
\end{align*}
\end{theo}

\begin{theo}[$(5,5\mFr+3)$-Curves]\label{T:C55m3}
Let a non-degenerate $(5,5\mFr+3)$-curve of genus $g=10\mFr+4$ be defined by
\begin{align}\label{V55m3Eq}
-y^5 + x^{5\mFr+3} + & y^3 \sum_{i=0}^{2\mFr+1} \lambda_{5i+1} x^{2\mFr+1-i} 
+ y^2 \sum_{i=0}^{3\mFr+1} \lambda_{5i+4} x^{3\mFr+1-i} \\
&+ y \sum_{i=0}^{4\mFr+2} \lambda_{5i+2} x^{4\mFr+2-i} 
+ \sum_{i=1}^{5\mFr +2} \lambda_{5i+5} x^{5\mFr+2-i} = 0. \notag
\end{align}
Let $u = \mathcal{A}(D)$ be the Abel image of  a non-special degree $g$ divisor  $D$
on the curve. Then $D$ is uniquely defined by the system of equations 
\begin{subequations}\label{REqsC55m3}
\begin{gather}
\mathcal{R}_{20\mFr+8 + l}(x,y;u)=0,\qquad l=0,1,2,3
\end{gather}
with four entire rational functions of the weights $2g + l =20\mFr + 8 + l$, $l=0,1,2,3$:
\begin{align}
&\mathcal{R}_{20\mFr+8}(x,y;u) 
= y x^{3\mFr+1} - \sum_{i=1}^{10\mFr+4} \wp_{1,\wFr_{i}}(u) \M_{-\wFr_{i}}, \\
&\mathcal{R}_{20\mFr+9}(x,y;u) 
= 2 y^3 x^{\mFr}  - \lambda_1 y x^{3\mFr+1} \\
&\phantom{\mathcal{R}_{20\mFr+9}(x,y;u)} \quad 
- \sum_{i=1}^{10\mFr+4} \big( \wp_{2,\wFr_{i}}(u) - \wp_{1,1,\wFr_{i}}(u) \big)  \M_{-\wFr_{i}}, \notag \\
&\mathcal{R}_{20\mFr+10}(x,y;u) 
= 3 x^{4\mFr+2} + \lambda_1 y^3 x^{\mFr} + 2 \lambda_2 y x^{3\mFr+1} \\
&\quad - \sum_{i=1}^{10\mFr+4} \big(\wp_{3,\wFr_{i}}(u) - \tfrac{3}{2} \wp_{1,2,\wFr_{i}}(u)
- \tfrac{1}{2} \lambda_1 \wp_{1,1,\wFr_{i}}(u)
 + \tfrac{1}{2} \wp_{1,1,1,\wFr_{i}}(u) \big)  \M_{-\wFr_{i}}, \notag \\
&\mathcal{R}_{20\mFr+11}(x,y;u) 
= 4 y^2 x^{2\mFr+1} - 2 \lambda_1 x^{4\mFr+2} + 2 \lambda_2 y^3 x^{\mFr} - \lambda_1 \lambda_2 y x^{3\mFr+1} \\
&\quad - \sum_{i=1}^{10\mFr+4} \big(\wp_{4,\wFr_{i}}(u) - \tfrac{1}{2} \wp_{2,2,\wFr_{i}}(u)
- \tfrac{4}{3} \wp_{1,3,\wFr_{i}}(u) + \tfrac{2}{3} \lambda_1 \wp_{1,2,\wFr_{i}}(u) \notag \\
&\quad  - \tfrac{1}{3}  (\lambda_2 - \tfrac{1}{2} \lambda_1^2 ) \wp_{1,1,\wFr_{i}}(u) 
+ \wp_{1,1,2,\wFr_{i}}(u) - \tfrac{1}{6}  \wp_{1,1,1,1,\wFr_{i}}(u) \big)  \M_{-\wFr_{i}}, \notag  
 \end{align}
\end{subequations}
where
\begin{align*}
&\M_{-(5i-2)} = y^3 x^{\mFr-i},\quad  i=1,\, \dots,\, \mFr, \\
&\M_{-(5i-4)} = y^2 x^{2\mFr+1-i},\quad  i=1,\, \dots,\, 2\mFr+1, \\
&\M_{-(5i-1)} = y x^{3\mFr+1-i},\quad  i=1,\, \dots,\, 3\mFr+1, \\
&\M_{-(5i-3)} = x^{4\mFr+2-i},\quad  i=1,\, \dots,\, 4\mFr+2.
\end{align*}
\end{theo}

\begin{theo}[$(5,5\mFr+4)$-Curves]\label{T:C55m4}
Let a non-degenerate $(5,5\mFr+4)$-curve of genus $g=10\mFr+6$ be defined by
\begin{align}\label{V55m4Eq}
-y^5 + x^{5\mFr+4} + & y^3 \sum_{i=0}^{2\mFr+1} \lambda_{5i+3} x^{2\mFr+1-i} 
+ y^2 \sum_{i=0}^{3\mFr+2} \lambda_{5i+2} x^{3\mFr+2-i} \\
&+ y \sum_{i=0}^{4\mFr+3} \lambda_{5i+1} x^{4\mFr+3-i} 
+ \sum_{i=1}^{5\mFr +3} \lambda_{5i+5} x^{5\mFr+3-i} = 0. \notag
\end{align}
Let $u = \mathcal{A}(D)$ be the Abel image of  a non-special degree $g$ divisor  $D$
on the curve. Then $D$ is uniquely defined by the system of equations 
\begin{subequations}\label{REqsC55m4}
\begin{gather}
\mathcal{R}_{20\mFr+12 + l}(x,y;u)=0,\qquad l=0,1,2,3
\end{gather}
with four entire rational functions of the weights $2g + l =20\mFr +12 + l$, $l=0,1,2,3$:
\begin{align}
\mathcal{R}_{20\mFr+12}(x,y;u) 
&= y^3 x^{\mFr} - \sum_{i=1}^{10\mFr+6} \wp_{1,\wFr_{i}}(u) \M_{-\wFr_{i}}, \\
\mathcal{R}_{20\mFr+13}(x,y;u) 
&= 2 y^2 x^{2\mFr+1}  + \lambda_1 y^3 x^{\mFr} \\
&\quad - \sum_{i=1}^{10\mFr+6} \big( \wp_{2,\wFr_{i}}(u) - \wp_{1,1,\wFr_{i}}(u) \big)  \M_{-\wFr_{i}}, \notag \\
\mathcal{R}_{20\mFr+14}(x,y;u) 
 &= 3 y x^{3\mFr+2} + 2 \lambda_1 y^2 x^{2\mFr+1} + \lambda_2 y^3 x^{\mFr} \\
&\hspace{-25mm} - \sum_{i=1}^{10\mFr+6} \big(\wp_{3,\wFr_{i}}(u) - \tfrac{3}{2} \wp_{1,2,\wFr_{i}}(u)
+ \tfrac{1}{2} \lambda_1 \wp_{1,1,\wFr_{i}}(u)
 + \tfrac{1}{2} \wp_{1,1,1,\wFr_{i}}(u) \big)  \M_{-\wFr_{i}}, \notag \\
\mathcal{R}_{20\mFr+15}(x,y;u) 
 &= 4 x^{4\mFr+3} + 3 \lambda_1 y x^{3\mFr+2} + 2 \lambda_2 y^2 x^{2\mFr+1} + \lambda_3 y^3 x^{\mFr}  \\
&\hspace{-25mm} - \sum_{i=1}^{10\mFr+6} \big(\wp_{4,\wFr_{i}}(u) - \tfrac{1}{2} \wp_{2,2,\wFr_{i}}(u)
- \tfrac{4}{3} \wp_{1,3,\wFr_{i}}(u) + \tfrac{5}{6} \lambda_1 \wp_{1,2,\wFr_{i}}(u) \notag \\
& \hspace{-10mm}  + \tfrac{1}{3} (\lambda_2 -  \lambda_1^2 ) \wp_{1,1,\wFr_{i}}(u)
+ \wp_{1,1,2,\wFr_{i}}(u) - \tfrac{1}{2} \lambda_1 \wp_{1,1,1,\wFr_{i}}(u) \notag  \\
& \hspace{35mm}   - \tfrac{1}{6}  \wp_{1,1,1,1,\wFr_{i}}(u) \big)  \M_{-\wFr_{i}}, \notag  
 \end{align}
\end{subequations}
where
\begin{align*}
&\M_{-(5i-1)} = y^3 x^{\mFr-i},\quad  i=1,\, \dots,\, \mFr, \\
&\M_{-(5i-2)} = y^2 x^{2\mFr+1-i},\quad  i=1,\, \dots,\, 2\mFr+1, \\
&\M_{-(5i-3)} = y x^{3\mFr+2-i},\quad  i=1,\, \dots,\, 3\mFr+2, \\
&\M_{-(5i-4)} = x^{4\mFr+3-i},\quad  i=1,\, \dots,\, 4\mFr+3.
\end{align*}
\end{theo}

\subsection{Proof of Theorem~\ref{T:C55m1} ($(5,5\mFr+1)$-Curves)}
We use the following parameterization of \eqref{V55m1Eq}
\begin{equation*}
 x(\xi) = \xi^{-5},\quad y(\xi) = \xi^{-5\mFr-1}\bigg(1 + \frac{\lambda_2}{5} \xi^2 
 + \frac{\lambda_3}{5} \xi^3 + \Big(\frac{\lambda_4}{5}  + \frac{\lambda_2^2}{5^2} \Big)\xi^4 
 + O(\xi^5)\bigg).
\end{equation*}
The basis of differentials of the first kind has the form
\begin{align*}
&\rmd u_{5i-4} = y^3 x^{\mFr-i} \frac{\rmd x}{\partial_y f}  
= \xi^{5i-5}  \Big(1 + \frac{2\lambda_2}{5} \xi^2 + \frac{\lambda_3}{5} \xi^3 + O(\xi^5)\Big) \rmd \xi,\quad i = 1,\, \dots,\, \mFr,\\
&\rmd u_{5i-3}= y^2 x^{2\mFr-i} \frac{\rmd x}{\partial_y f}  
=\xi^{5i-4}  \Big(1 + \frac{\lambda_2}{5} \xi^2 + O(\xi^4)\Big) \rmd \xi,\quad i = 1,\, \dots,\, 2\mFr,\\
&\rmd u_{5i-2}= y x^{3\mFr-i} \frac{\rmd x}{\partial_y f}  
=\xi^{5i-3}  \Big(1 + O(\xi^3)\Big) \rmd \xi,\quad i = 1,\, \dots,\, 3\mFr,\\
&\rmd u_{5i-1}= x^{4\mFr-i} \frac{\rmd x}{\partial_y f}  
=\xi^{5i-2}  \Big(1   + O(\xi^2)\Big) \rmd \xi,
\quad i = 1,\, \dots,\, 4\mFr.
\end{align*}
By integration with respect to $\xi$ we find the integrals of the first kind
\begin{align*}
\mathcal{A}(\xi) = \big(& u_1(\xi),\, u_2(\xi),\, u_3(\xi),\, u_4(\xi),\, \dots,\, 
u_{5\mFr-4}(\xi),\, u_{5\mFr-3}(\xi),\, \\ 
&u_{5\mFr-2}(\xi),\, u_{5\mFr-1}(\xi),\, \dots,\,  u_{10\mFr-3}(\xi),\, u_{10\mFr-2}(\xi),\, 
u_{10\mFr-1}(\xi),\, \dots,\, \\
& u_{15\mFr-2}(\xi),\, u_{15\mFr-1}(\xi),\,
\dots,\, u_{20\mFr-6}(\xi),\, u_{20\mFr-1}(\xi) \big)^t,
\end{align*} 
and then
\begin{align}\label{DA0C55m1}
\mathcal{A}(0) = 0,\quad  \mathcal{A}'(0) = (\delta_{i,1}),\quad  \mathcal{A}''(0) = (\delta_{i,2}),\quad  
\mathcal{A}^{(3)}(0) = \big(2\delta_{i,3} + \tfrac{4}{5}\lambda_2 \delta_{i,1} \big).
\end{align} 

Let the four differentials of the second kind of the smallest Sat\={o} weights be
\begin{align}\label{Diff2DefsS}
&\big(\rmd r_1,\, \rmd r_2,\, \rmd r_3,\, \rmd r_4 \big) = 
\frac{\rmd x}{\partial_y f} 
 \big(x^{4\mFr},\, 2 yx^{3\mFr},\, 3 y^2 x^{2\mFr},\, 4 y^3 x^{\mFr}\big). 
\end{align}
We start from the simplest form of these differentials. 
Then we apply the condition \eqref{rCond} to expansions of
\eqref{Diff2DefsS} near infinity, and find the differentials of the second kind
associated with differentials of the first kind:
\begin{align*}
&\begin{pmatrix} 
\rmd r_1 \\ \rmd r_2 \\ \rmd \widetilde{r}_3 \\ \rmd \widetilde{r}_4
\end{pmatrix} = 
\frac{\rmd x}{\partial_y f} 
\begin{pmatrix} 
x^{4\mFr} \\ 2 yx^{3\mFr} \\ 3 y^2 x^{2\mFr} - \lambda_2 x^{4\mFr} \\
 4 y^3 x^{\mFr} - 2\lambda_2 yx^{3\mFr} - \lambda_3 x^{4\mFr}
 \end{pmatrix}.
\end{align*}
By integration with respect to $\xi$ we obtain
\begin{align*}
& r_1(\xi) = -\xi^{-1}  + O(\xi) + c_1,\\
& r_2(\xi) = - \xi^{-2}  + O(\xi) + c_2,\\
& \widetilde{r}_3(\xi) = - \xi^{-3} + \frac{2\lambda_2}{5} \xi^{-1} + O(\xi) + c_3,\\
& \widetilde{r}_4(\xi) = - \xi^{-4} + \frac{\lambda_2}{5} \xi^{-2}  + \frac{\lambda_3}{5} \xi^{-1} + O(\xi) + c_4.
\end{align*}
Using \eqref{DA0C55m1}, we find
\begin{multline*}
\frac{\rmd}{\rmd \xi} \log \sigma\big(u - \mathcal{A}(\xi)\big) = - \zeta_1(u)
- \big(\zeta_2(u) + \wp_{1,1}(u)\big) \xi \\
- \big(\zeta_3(u) + \tfrac{2}{5}\lambda_2 \zeta_1(u) + \tfrac{3}{2} \wp_{1,2}(u) 
- \tfrac{1}{2} \wp_{1,1,1}(u) \big) \xi^2 
- \big(\zeta_4(u) + \tfrac{1}{5}\lambda_2 \zeta_2(u) + \tfrac{1}{5}\lambda_3 \zeta_1(u) \\
 + \tfrac{1}{2} \wp_{2,2}(u) + \tfrac{4}{3} \wp_{1,3}(u)  + \tfrac{8}{15} \lambda_2 \wp_{1,1}(u) 
 - \wp_{1,1,2}(u) + \tfrac{1}{6} \wp_{1,1,1,1}(u) \big)\xi^3+ O(\xi^4),
\end{multline*}
where $u = \mathcal{A}(D)$ is the Abel's image of a non-special divisor 
$D = \sum_{k=1}^{10\mFr} (x_k,y_k)$. Then we construct four
equations of the form \eqref{rExpr}:
\begin{align}\label{ZetaC55m1}
& \sum_{k=1}^{10\mFr} \begin{pmatrix}
r_1(x_k,\,y_k) \\ r_2(x_k,\,y_k) \\ \widetilde{r}_3(x_k,\,y_k) \\ \widetilde{r}_4(x_k,\,y_k) \\ 
\phantom{\tfrac{4}{3} \big(\big)} 
\end{pmatrix} = - \begin{pmatrix} 
\zeta_1(u) \\
\zeta_2(u) + \wp_{1,1}(u) \\
\zeta_3(u) + \tfrac{3}{2} \wp_{1,2}(u) - \tfrac{1}{2} \wp_{1,1,1}(u) \\
\zeta_4(u)  + \tfrac{1}{2} \wp_{2,2}(u) + \tfrac{4}{3} \wp_{1,3}(u) + \tfrac{1}{3} \lambda_2 \wp_{1,1}(u) \\
 - \wp_{1,1,2}(u) + \tfrac{1}{6} \wp_{1,1,1,1}(u) \end{pmatrix}.
\end{align}
Differentiating all equations with respect to $x_1$, we obtain \eqref{REqsC55m1}. $\qede$

\subsection{Proof of Theorem~\ref{T:C55m2} ($(5,5\mFr+2)$-Curves)}
We use the following parameterization of \eqref{V55m2Eq}
\begin{equation*}
\begin{split}
 x(\xi) &= \xi^{-5},\\ y(\xi) &= \xi^{-5\mFr-2} \bigg(1 + \frac{\lambda_1}{5} \xi 
 + \Big(\frac{\lambda_3}{5} - \frac{\lambda_1^3}{5^3} \Big) \xi^3 
 + \Big(\frac{\lambda_4}{5}  - \frac{\lambda_1 \lambda_3}{5^2}   + 
 \frac{\lambda_1^4}{5^4} \Big)\xi^4 
 + O(\xi^5)\bigg).
 \end{split}
\end{equation*}
The basis of differentials of the first kind has the form
\begin{align*}
&\rmd u_{5i-3} = y^3 x^{\mFr-i} \frac{\rmd x}{\partial_y f}  
= \xi^{5i-4}  \Big(1 + \frac{\lambda_1}{5} \xi - \frac{3 \lambda_1^2}{5^2} \xi^2  + O(\xi^4)\Big) \rmd \xi,
\quad i = 1,\, \dots,\, \mFr,\\
&\rmd u_{5i-1}= y^2 x^{2\mFr-i} \frac{\rmd x}{\partial_y f}  
=\xi^{5i-2}  \Big(1  + O(\xi^2)\Big) \rmd \xi,\quad i = 1,\, \dots,\, 2\mFr,\\
&\rmd u_{5i-4}= y x^{3\mFr+1-i} \frac{\rmd x}{\partial_y f}  
=\xi^{5i-5}  \Big(1 - \frac{\lambda_1}{5} \xi - \frac{2 \lambda_1^2}{5^2} \xi^2 
- \Big(\frac{2\lambda_3}{5} - \frac{7\lambda_1^3}{5^3}\Big) \xi^3+ O(\xi^5)\Big) \rmd \xi, \\
&\qquad  i = 1,\, \dots,\, 3\mFr+1,\\
&\rmd u_{5i-2}= x^{4\mFr+1-i} \frac{\rmd x}{\partial_y f}  
=\xi^{5i-3}  \Big(1 - \frac{2\lambda_1}{5} \xi  + O(\xi^3)\Big) \rmd \xi,
\quad  i = 1,\, \dots,\, 4\mFr+1.
\end{align*}
By integration with respect to $\xi$ we find the integrals of the first kind
\begin{align*}
\mathcal{A}(\xi) = \big(& u_1(\xi),\, u_2(\xi),\, u_3(\xi),\, u_4(\xi),\, \dots,\, 
u_{5\mFr-3}(\xi),\, u_{5\mFr-2}(\xi),\, \\ 
&u_{5\mFr-1}(\xi),\, u_{5\mFr+1}(\xi),\, \dots,\,  u_{10\mFr-1}(\xi),\, u_{10\mFr+1}(\xi),\, 
u_{10\mFr+3}(\xi),\, \dots,\, \\
& u_{15\mFr+1}(\xi),\, u_{15\mFr+3}(\xi),\,
\dots,\, u_{20\mFr-2}(\xi),\, u_{20\mFr+3}(\xi) \big)^t,
\end{align*} 
and then
\begin{gather}\label{DA0C55m2}
\begin{split}
&\mathcal{A}(0) = 0,\quad  \mathcal{A}'(0) = (\delta_{i,1}),\quad  
\mathcal{A}''(0) = (\delta_{i,2} -  \tfrac{1}{5} \lambda_1 \delta_{i,1}),\\  
&\mathcal{A}^{(3)}(0) = \big( 2\delta_{i,3} + \tfrac{2}{5}\lambda_1 \delta_{i,2} - (\tfrac{2}{5})^2 \lambda_1^2 \delta_{i,1}   \big).
\end{split}
\end{gather} 

Let the four differentials of the second kind of weights $1$, $2$, $3$, $4$ be:
\begin{align*}
&\big(\rmd r_1,\, \rmd r_2,\, \rmd r_3,\, \rmd r_4 \big) = 
\frac{\rmd x}{\partial_y f} 
 \big(y^2 x^{2\mFr},\, 2 x^{4\mFr+1},\, 3 y^3 x^{\mFr},\, 4 y x^{3\mFr+1}\big).
\end{align*}
With the help of \eqref{rCond} we find the differentials of the second kind associated with 
differentials of the first kind:
\begin{align*}
&\begin{pmatrix} 
\rmd r_1 \\ \rmd \widetilde{r}_2 \\ \rmd \widetilde{r}_3 \\ \rmd \widetilde{r}_4
\end{pmatrix} = 
\frac{\rmd x}{\partial_y f} 
\begin{pmatrix} 
y^2 x^{2\mFr} \\ 
2 x^{4\mFr+1} + \lambda_1 y^2 x^{2\mFr} \\ 
3 y^3 x^{\mFr} - \lambda_1 x^{4\mFr+1} \\
4 y x^{3\mFr+1} + 2 \lambda_1 y^3 x^{\mFr} + 2 \lambda_3 y^2 x^{2\mFr}
 \end{pmatrix}.
\end{align*}
By integration with respect to $\xi$ we obtain
\begin{align*}
& r_1(\xi) = -\xi^{-1}  + O(\xi) + c_1,\\
& \widetilde{r}_2(\xi) = - \xi^{-2} - \frac{\lambda_1}{5} \xi^{-1} + O(\xi) + c_2,\\
& \widetilde{r}_3(\xi) = - \xi^{-3} + \frac{\lambda_1}{5} \xi^{-2} - \frac{\lambda_1^2}{5^2} \xi^{-1} + O(\xi) + c_3,\\
& \widetilde{r}_4(\xi) = - \xi^{-4} - \frac{2\lambda_1}{15} \xi^{-3} - \frac{\lambda_1^2}{5^2} \xi^{-2}  
- 2\Big(\frac{\lambda_3}{5}  - \frac{\lambda_1^3}{5^3} \Big) \xi^{-1} + O(\xi) + c_4.
\end{align*}
Using \eqref{DA0C55m2}, we find
\begin{multline*}
\frac{\rmd}{\rmd \xi} \log \sigma\big(u - \mathcal{A}(\xi)\big) = - \zeta_1(u)
- \big(\zeta_2(u) - \tfrac{1}{5}\lambda_1 \zeta_1(u) + \wp_{1,1}(u)\big) \xi \\
- \big(\zeta_3(u) + \tfrac{1}{5}\lambda_1 \zeta_2(u) - \tfrac{2}{5^2}\lambda_1^2 \zeta_1(u) 
+ \tfrac{3}{2} \wp_{1,2}(u) - \tfrac{3}{10} \lambda_1  \wp_{1,1}(u) 
- \tfrac{1}{2} \wp_{1,1,1}(u) \big) \xi^2 \\
- \big(\zeta_4(u) - \tfrac{2}{5}\lambda_1 \zeta_3(u) - \tfrac{3}{5^2}\lambda_1^2 \zeta_2(u) 
- (\tfrac{2}{5}\lambda_3 - \tfrac{7}{5^3} \lambda_1^3) \zeta_1(u) 
 + \tfrac{1}{2} \wp_{2,2}(u) + \tfrac{4}{3} \wp_{1,3}(u) \\
 + \tfrac{1}{15} \lambda_1 \wp_{1,2}(u) - \tfrac{13}{6\cdot 5^2} \lambda_1^2 \wp_{1,1}(u) 
 - \wp_{1,1,2}(u) + \tfrac{1}{5} \lambda_1 \wp_{1,1,1}(u) + \tfrac{1}{6} \wp_{1,1,1,1}(u) \big)\xi^3+ O(\xi^4),
\end{multline*}
where $u = \mathcal{A}(D)$ is the Abel's image of a non-special divisor 
$D = \sum_{k=1}^{10\mFr+2} (x_k,y_k)$. Then we construct four
equations of the form \eqref{rExpr}:
\begin{align}\label{ZetaC55m2}
& \sum_{k=1}^{10\mFr+2} \begin{pmatrix}
r_1(x_k,\,y_k) \\ \widetilde{r}_2(x_k,\,y_k) \\ 
\widetilde{r}_3(x_k,\,y_k) \\ 
\widetilde{r}_4(x_k,\,y_k) \\ 
\phantom{\tfrac{4}{3} \big(\big)} 
\end{pmatrix} = - \begin{pmatrix} 
\zeta_1(u) \\
\zeta_2(u) + \wp_{1,1}(u) \\
\zeta_3(u) + \tfrac{3}{2} \wp_{1,2}(u) - \tfrac{1}{2}\lambda_1 \wp_{1,1}(u) 
-  \tfrac{1}{2} \wp_{1,1,1}(u)\\
\zeta_4(u) + \tfrac{1}{2} \wp_{2,2}(u) + \tfrac{4}{3} \wp_{1,3}(u) 
+ \tfrac{2}{3} \lambda_1 \wp_{1,2}(u)
 \\  - \tfrac{1}{6} \lambda_1^2 \wp_{1,1}(u) - \wp_{1,1,2}(u) 
+ \tfrac{1}{6} \wp_{1,1,1,1}(u) \end{pmatrix}.
\end{align}
Differentiating all equations with respect to $x_1$, we obtain \eqref{REqsC55m2}. 
$\qede$

\subsection{Proof of Theorem~\ref{T:C55m3} ($(5,5\mFr+3)$-Curves)}
We use the following parameterization of \eqref{V55m3Eq}
\begin{equation*}
\begin{split}
 x(\xi) &= \xi^{-5},\\ y(\xi) &= \xi^{-5\mFr-3} \bigg(1 + \frac{\lambda_1}{5} \xi 
 + \Big(\frac{\lambda_2}{5} + \frac{\lambda_1^2}{5^2} \Big) \xi^2 \\
&\qquad\qquad\quad  + \Big(\frac{\lambda_4}{5}  - \frac{\lambda_2^2}{5^2} - \frac{3}{5^3} \lambda_1^2 \lambda_2 
 - \frac{2}{5^4} \lambda_1^4 \Big)\xi^4 
 + O(\xi^5)\bigg).
 \end{split}
\end{equation*}
The basis of differentials of the first kind has the form
\begin{align*}
&\rmd u_{5i-2} = y^3 x^{\mFr-i} \frac{\rmd x}{\partial_y f}  
= \xi^{5i-3}  \Big(1 + \frac{2\lambda_1}{5} \xi  + O(\xi^3)\Big) \rmd \xi,
\quad i = 1,\, \dots,\, \mFr,\\
&\rmd u_{5i-4}= y^2 x^{2\mFr+1-i} \frac{\rmd x}{\partial_y f}  
=\xi^{5i-5}  \Big(1 +  \frac{\lambda_1}{5} \xi 
- \Big(\frac{\lambda_2}{5} + \frac{2\lambda_1^2}{5^2} \Big) \xi^2  \\ 
&\qquad\qquad 
- \Big(\frac{6\lambda_1 \lambda_2}{5^2} + \frac{7\lambda_1^3}{5^3}\Big) \xi^3 
+ O(\xi^5)\Big) \rmd \xi,\quad i = 1,\, \dots,\, 2\mFr+1,\\
&\rmd u_{5i-1} = y x^{3\mFr+1-i} \frac{\rmd x}{\partial_y f}  
=\xi^{5i-2}  \Big(1 + O(\xi^2)\Big) \rmd \xi,\quad  i = 1,\, \dots,\, 3\mFr+1,\\
&\rmd u_{5i-3}= x^{4\mFr+2-i} \frac{\rmd x}{\partial_y f}  
=\xi^{5i-4}  \Big(1 - \frac{\lambda_1}{5} \xi  
- 3\Big(\frac{\lambda_2}{5} + \frac{\lambda_1^2}{5^2}\Big) \xi^2 + O(\xi^4)\Big) \rmd \xi, \\ 
&\qquad  i = 1,\, \dots,\, 4\mFr+2.
\end{align*}
By integration with respect to $\xi$ we find the integrals of the first kind
\begin{align*}
\mathcal{A}(\xi) = \big(& u_1(\xi),\, u_2(\xi),\, u_3(\xi),\, u_4(\xi),\, \dots,\, 
u_{5\mFr-2}(\xi),\, u_{5\mFr-1}(\xi),\, \\ 
&u_{5\mFr+1}(\xi),\, u_{5\mFr+2}(\xi),\, \dots,\,  
u_{10\mFr+1}(\xi),\, u_{10\mFr+2}(\xi),\, u_{10\mFr+4}(\xi),\, \dots,\, \\
& u_{15\mFr+4}(\xi),\, u_{15\mFr+7}(\xi),\,
\dots,\, u_{20\mFr+2}(\xi),\, u_{20\mFr+7}(\xi) \big)^t,
\end{align*} 
and then
\begin{gather}\label{DA0C55m3}
\begin{split}
&\mathcal{A}(0) = 0,\quad  \mathcal{A}'(0) = (\delta_{i,1}),\quad  
\mathcal{A}''(0) = ( \delta_{i,2} + \tfrac{1}{5} \lambda_1 \delta_{i,1}),\\  
&\mathcal{A}^{(3)}(0) = \big( 2\delta_{i,3} - \tfrac{2}{5}\lambda_1 \delta_{i,2}
- \big( \tfrac{2}{5}\lambda_2 + (\tfrac{2}{5})^2 \lambda_1^2 \big) \delta_{i,1}  \big).
\end{split}
\end{gather} 
Let the four differentials of the second kind of weights $1$, $2$, $3$, $4$ be:
\begin{align*}
&\big(\rmd r_1,\, \rmd r_2,\, \rmd r_3,\, \rmd r_4 \big) = 
\frac{\rmd x}{\partial_y f} 
 \big(y x^{3\mFr+1},\, 2 y^3 x^{\mFr},\, 3 x^{4\mFr+2},\, 4 y^2 x^{2\mFr+1}\big). 
\end{align*}
With the help of \eqref{rCond} we find the differentials of the second kind associated with 
differentials of the first kind:
\begin{align*}
&\begin{pmatrix} 
\rmd r_1 \\ \rmd \widetilde{r}_2 \\ \rmd \widetilde{r}_3 \\ \rmd \widetilde{r}_4
\end{pmatrix} = 
\frac{\rmd x}{\partial_y f} 
\begin{pmatrix} 
y x^{3\mFr+1}\\ 
2 y^3 x^{\mFr} - \lambda_1 y x^{3\mFr+1} \\ 
3 x^{4\mFr+2} + \lambda_1 y^3 x^{\mFr} + 2\lambda_2 y x^{3\mFr+1} \\
4 y^2 x^{2\mFr+1} - 2 \lambda_1 x^{4\mFr+2} + 2 \lambda_2 y^3 x^{\mFr}
- \lambda_1 \lambda_2 y x^{3\mFr+1}
 \end{pmatrix}.
\end{align*}
By integration with respect to $\xi$ we obtain
\begin{align*}
& r_1(\xi) = -\xi^{-1}  + O(\xi) + c_1,\\
& \widetilde{r}_2(\xi) = - \xi^{-2} + \frac{\lambda_1}{5} \xi^{-1} + O(\xi) + c_2,\\
& \widetilde{r}_3(\xi) = - \xi^{-3} - \frac{\lambda_1}{5} \xi^{-2} - \Big( \frac{\lambda_2}{5} 
+ \frac{\lambda_1^2}{5^2} \Big) \xi^{-1} + O(\xi) + c_3,\\
& \widetilde{r}_4(\xi) = - \xi^{-4} + \frac{2  \lambda_1}{5} \xi^{-3} 
- \Big( \frac{3\lambda_2}{5}  + \frac{\lambda_1^2 }{5^2} \Big) \xi^{-2} 
- \Big(\frac{\lambda_2 \lambda_1}{5^2}  + \frac{2\lambda_1^3}{5^3} \Big) \xi^{-1} + O(\xi) + c_4.
\end{align*}
Using \eqref{DA0C55m3}, we find
\begin{multline*}
\frac{\rmd}{\rmd \xi} \log \sigma\big(u - \mathcal{A}(\xi)\big) = - \zeta_1(u)
- \big(\zeta_2(u) + \tfrac{1}{5}\lambda_1 \zeta_1(u) + \wp_{1,1}(u)\big) \xi \\
- \big(\zeta_3(u) - \tfrac{1}{5}\lambda_1 \zeta_2(u) - (\tfrac{1}{5}\lambda_2 + \tfrac{2}{5^2}\lambda_1^2) \zeta_1(u) 
+ \tfrac{3}{2} \wp_{1,2}(u) + \tfrac{3}{10} \lambda_1  \wp_{1,1}(u) 
- \tfrac{1}{2} \wp_{1,1,1}(u) \big) \xi^2 \\
- \big(\zeta_4(u) + \tfrac{2}{5}\lambda_1 \zeta_3(u) - 3(\tfrac{1}{5} \lambda_2 + \tfrac{1}{5^2}\lambda_1^2) \zeta_2(u) 
- (\tfrac{6}{5^2}\lambda_2 \lambda_1 + \tfrac{7}{5^3} \lambda_1^3) \zeta_1(u) \\
 + \tfrac{1}{2} \wp_{2,2}(u) + \tfrac{4}{3} \wp_{1,3}(u) 
 - \tfrac{1}{15} \lambda_1 \wp_{1,2}(u) - \tfrac{1}{3}(\tfrac{4}{5}\lambda_2 + \tfrac{13}{2\cdot 5^2} \lambda_1^2) \wp_{1,1}(u) \\
 - \wp_{1,1,2}(u) - \tfrac{1}{5} \lambda_1 \wp_{1,1,1}(u) + \tfrac{1}{6} \wp_{1,1,1,1}(u) \big)\xi^3+ O(\xi^4),
\end{multline*}
where $u = \mathcal{A}(D)$ is the Abel's image of a non-special divisor 
$D = \sum_{k=1}^{10\mFr+4} (x_k,y_k)$. Then we construct four
equations of the form \eqref{rExpr}:
\begin{align}\label{ZetaC55m3}
& \sum_{k=1}^{10\mFr+4} \begin{pmatrix}
r_1(x_k,\,y_k) \\ \widetilde{r}_2(x_k,\,y_k) \\ 
\widetilde{r}_3(x_k,\,y_k) \\
\widetilde{r}_4(x_k,\,y_k) \\ 
\phantom{\tfrac{4}{3} \big(\big)} 
\end{pmatrix} = - \begin{pmatrix} 
\zeta_1(u) \\
\zeta_2(u) + \wp_{1,1}(u) \\
\zeta_3(u) + \tfrac{3}{2} \wp_{1,2}(u) + \tfrac{1}{2} \lambda_1 \wp_{1,1}(u)  - \tfrac{1}{2} \wp_{1,1,1}(u) \\
\zeta_4(u)  + \tfrac{1}{2} \wp_{2,2}(u) + \tfrac{4}{3} \wp_{1,3}(u)  - \tfrac{2}{3} \lambda_1 \wp_{1,2}(u) \\
+  \tfrac{1}{3} (\lambda_2 - \tfrac{1}{2} \lambda_1^2 ) \wp_{1,1}(u) 
- \wp_{1,1,2}(u) + \tfrac{1}{6} \wp_{1,1,1,1}(u) \end{pmatrix}.
\end{align}
Differentiating all equations with respect to $x_1$, we obtain \eqref{REqsC55m3}. 
$\qede$

\subsection{Proof of Theorem~\ref{T:C55m4} ($(5,5\mFr+4)$-Curves)}
We use the following parameterization of \eqref{V55m4Eq}
\begin{equation*}
\begin{split}
 x(\xi) &= \xi^{-5},\\ y(\xi) &= \xi^{-5\mFr-4} \bigg(1 + \frac{\lambda_1}{5} \xi 
 + \Big(\frac{\lambda_2}{5} - \frac{\lambda_1^2}{5^2} \Big) \xi^2 
 + \Big(\frac{\lambda_3}{5}  - \frac{\lambda_2\lambda_1}{5^2} + \frac{\lambda_1^3}{5^3} \Big)\xi^3 
 + O(\xi^5)\bigg).
 \end{split}
\end{equation*}
The basis of differentials of the first kind has the form
\begin{align*}
&\rmd u_{5i-1} = y^3 x^{\mFr-i} \frac{\rmd x}{\partial_y f}  
= \xi^{5i-2}  \Big(1 + O(\xi^2)\Big) \rmd \xi,\quad  i = 1,\, \dots,\, \mFr,\\
&\rmd u_{5i-2}= y^2 x^{2\mFr+1-i} \frac{\rmd x}{\partial_y f}  
=\xi^{5i-3}  \Big(1 -  \frac{\lambda_1}{5} \xi  + O(\xi^3)\Big) \rmd \xi,\quad i = 1,\, \dots,\, 2\mFr+1,\\
&\rmd u_{5i-3} = y x^{3\mFr+2-i} \frac{\rmd x}{\partial_y f}  
=\xi^{5i-4}  \Big(1 - \frac{2\lambda_1}{5} \xi 
- \Big(\frac{\lambda_2}{5} - \frac{3 \lambda_1^2}{5^2} \Big)\xi^2 + O(\xi^4)\Big) \rmd \xi, \\
&\qquad   i = 1,\, \dots,\, 3\mFr+2,\\
&\rmd u_{5i-4}= x^{4\mFr+3-i} \frac{\rmd x}{\partial_y f}  
=\xi^{5i-5}  \Big(1 - \frac{3\lambda_1}{5} \xi  
- \Big(\frac{2\lambda_2}{5} - \frac{7\lambda_1^2}{5^2} \Big) \xi^2 \\
&\qquad\qquad - \Big(\frac{\lambda_3}{5} - \frac{6\lambda_1\lambda_2 }{5^2}+ \frac{11 \lambda_1^3}{5^3} \Big) \xi^3+ O(\xi^5)\Big) \rmd \xi, \quad i = 1,\, \dots,\, 4\mFr+3.
\end{align*}
By integration over $\xi$ we find the integrals of the first kind
\begin{align*}
\mathcal{A}(\xi) = \big(& u_1(\xi),\, u_2(\xi),\, u_3(\xi),\, u_4(\xi),\, \dots,\, 
u_{5\mFr-1}(\xi),\, u_{5\mFr+1}(\xi),\, \\ 
&u_{5\mFr+2}(\xi),\, u_{5\mFr+3}(\xi),\, \dots,\,  u_{10\mFr+3}(\xi),\, u_{10\mFr+6}(\xi),\, 
u_{10\mFr+7}(\xi),\, \dots,\, \\
& u_{15\mFr+7}(\xi),\, u_{15\mFr+11}(\xi),\,
\dots,\, u_{20\mFr+6}(\xi),\, u_{20\mFr+11}(\xi) \big)^t,
\end{align*} 
and then
\begin{gather}\label{DA0C55m3}
\begin{split}
&\mathcal{A}(0) = 0,\quad  \mathcal{A}'(0) = (\delta_{i,1}),\quad  
\mathcal{A}''(0) = (\delta_{i,2} - \tfrac{3}{5} \lambda_1 \delta_{i,1}),\\  
&\mathcal{A}^{(3)}(0) = \big(2\delta_{i,3} - \tfrac{4}{5}\lambda_1 \delta_{i,2}
- \big( \tfrac{4}{5}\lambda_2 - \tfrac{14}{5^2} \lambda_1^2 \big) \delta_{i,1}  \big).
\end{split}
\end{gather} 
Let the four differentials of the second kind of weights $1$, $2$, $3$, $4$ be:
\begin{align*}
&\big(\rmd r_1,\, \rmd r_2,\, \rmd r_3,\, \rmd r_4 \big) = 
\frac{\rmd x}{\partial_y f} 
 \big(y^3 x^{\mFr},\, 2 y^2 x^{2\mFr+1},\, 3 y x^{3\mFr+2},\, 4 x^{4\mFr+3}\big). 
\end{align*}
With the help of \eqref{rCond} we find the differentials of the second kind associated with 
differentials of the first kind:
\begin{align*}
&\begin{pmatrix} 
\rmd r_1 \\ \rmd \widetilde{r}_2 \\ \rmd \widetilde{r}_3 \\ \rmd \widetilde{r}_4
\end{pmatrix} = 
\frac{\rmd x}{\partial_y f} 
\begin{pmatrix} 
y^3 x^{\mFr} \\ 
2 y^2 x^{2\mFr+1} + \lambda_1 y^3 x^{\mFr} \\ 
3 y x^{3\mFr+2}+ 2 \lambda_1 y^2 x^{2\mFr+1} + \lambda_2 y^3 x^{\mFr} \\
4 x^{4\mFr+3} + 3 \lambda_1 y x^{3\mFr+2} + 2 \lambda_2 y^2 x^{2\mFr+1} + \lambda_3 y^3 x^{\mFr}
 \end{pmatrix}.
\end{align*}
By integration with respect to $\xi$ we obtain
\begin{align*}
& r_1(\xi) = -\xi^{-1}  + O(\xi) + c_1,\\
& \widetilde{r}_2(\xi) = - \xi^{-2} - \frac{3\lambda_1}{5} \xi^{-1} + O(\xi) + c_2,\\
& \widetilde{r}_3(\xi) = - \xi^{-3} - \frac{2\lambda_1}{5} \xi^{-2} - \Big( \frac{2\lambda_2}{5} 
- \frac{\lambda_1^2}{5^2} \Big) \xi^{-1} + O(\xi) + c_3,\\
& \widetilde{r}_4(\xi) = - \xi^{-4} - \frac{\lambda_1}{5}  \xi^{-3} 
- \Big( \frac{\lambda_2}{5} - \frac{\lambda_1^2}{5^2} \Big) \xi^{-2}  
- \Big(\frac{\lambda_3}{5} - \frac{\lambda_2 \lambda_1}{5^2}  
+ \frac{\lambda_1^3}{5^3} \big) \xi^{-1} + O(\xi) + c_4.
\end{align*}
Using \eqref{DA0C55m3}, we find
\begin{multline*}
\frac{\rmd}{\rmd \xi} \log \sigma\big(u - \mathcal{A}(\xi)\big) = - \zeta_1(u)
- \big(\zeta_2(u) - \tfrac{3}{5}\lambda_1 \zeta_1(u) + \wp_{1,1}(u)\big) \xi \\
- \big(\zeta_3(u) - \tfrac{2}{5}\lambda_1 \zeta_2(u) - (\tfrac{2}{5}\lambda_2 - \tfrac{7}{5^2}\lambda_1^2) \zeta_1(u) 
+ \tfrac{3}{2} \wp_{1,2}(u) - \tfrac{9}{10} \lambda_1  \wp_{1,1}(u) 
- \tfrac{1}{2} \wp_{1,1,1}(u) \big) \xi^2 \\
- \big(\zeta_4(u) - \tfrac{1}{5}\lambda_1 \zeta_3(u) - (\tfrac{1}{5} \lambda_2 - \tfrac{3}{5^2}\lambda_1^2) \zeta_2(u) 
- (\tfrac{1}{5} \lambda_3 - \tfrac{6}{5^2}\lambda_2 \lambda_1 + \tfrac{11}{5^3} \lambda_1^3) \zeta_1(u) \\
 + \tfrac{1}{2} \wp_{2,2}(u) + \tfrac{4}{3} \wp_{1,3}(u) 
 - \tfrac{17}{15} \lambda_1 \wp_{1,2}(u) - \tfrac{1}{3} (\tfrac{8}{5}\lambda_2 - \tfrac{83}{2\cdot 5^2} \lambda_1^2) \wp_{1,1}(u) \\
 - \wp_{1,1,2}(u) + \tfrac{3}{5} \lambda_1 \wp_{1,1,1}(u) + \tfrac{1}{6} \wp_{1,1,1,1}(u) \big)\xi^3+ O(\xi^4),
\end{multline*}
where $u = \mathcal{A}(D)$ is the Abel's image of a non-special divisor 
$D = \sum_{k=1}^{10\mFr+6} (x_k,y_k)$. Then we construct four
equations of the form \eqref{rExpr}:
\begin{align}\label{ZetaC55m4}
& \sum_{k=1}^{10\mFr+6} \begin{pmatrix}
r_1(x_k,\,y_k) \\ \widetilde{r}_2(x_k,\,y_k) \\ 
\widetilde{r}_3(x_k,\,y_k) \\ 
\widetilde{r}_4(x_k,\,y_k) \\ 
\phantom{\tfrac{4}{3} \big(\big)} \\ 
\phantom{\tfrac{4}{3} \big(\big)}
\end{pmatrix} = - \begin{pmatrix} 
\zeta_1(u) \\
\zeta_2(u)  + \wp_{1,1}(u) \\
\zeta_3(u) + \tfrac{3}{2} \wp_{1,2}(u) - \tfrac{1}{2} \lambda_1 \wp_{1,1}(u) - \tfrac{1}{2} \wp_{1,1,1}(u) \\
\zeta_4(u) + \tfrac{1}{2} \wp_{2,2}(u) + \tfrac{4}{3} \wp_{1,3}(u) - \tfrac{5}{6} \lambda_1 \wp_{1,2}(u) \\
- \tfrac{1}{3}(\lambda_2 - \lambda_1^2) \wp_{1,1}(u) \\  - \wp_{1,1,2}(u) + \tfrac{1}{2} \lambda_1 \wp_{1,1,1}(u)
+ \tfrac{1}{6} \wp_{1,1,1,1}(u) \end{pmatrix}.
\end{align}
Differentiating all equations with respect to $x_1$, we obtain \eqref{REqsC55m4}. 
$\qede$

\section{Conclusions}
The Jacobi inversion problem is solved in general by Theorems~\ref{T1} and \ref{T2},
where an $(n,s)$-curve is supposed to be non-hyperelliptic. However, the same approach works for
hyperelliptic curves, as explained in Example~\ref{E:HypC}.
Theorems~\ref{T:C33m1} and \ref{T:C33m2} give solutions on trigonal curves
of the types $(3,3\mFr+1)$ and $(3,3\mFr+2)$, $\mFr$
is a natural number. Theorems~\ref{T:C44m1} and \ref{T:C44m3} 
give solutions on tetragonal curves of the types $(4,4\mFr+1)$ and $(4,4\mFr+3)$.
Theorems~\ref{T:C55m1}, \ref{T:C55m2}, \ref{T:C55m3}, and \ref{T:C55m4}
 give solutions on pentagonal curves of the types $(5,5\mFr+1)$, $(5,5\mFr+2)$, $(5,5\mFr+3)$, and 
 $(5,5\mFr+4)$.
 The Jacobi inversion problem is solved in terms of multiply periodic $\wp$ functions
defined by \eqref{wp2Def}, \eqref{wp3Def} from the multivariable sigma function $\sigma$
of a curve under consideration.


\end{document}